\documentclass[a4paper]{article}

\usepackage[vlined,boxed,ruled,algosection]{algorithm2e} %Algorithm presentation package
\usepackage{amsmath}
\usepackage{latexsym}
\usepackage{amssymb}
\usepackage{graphicx}
\usepackage{amsthm}

\newtheorem{thm}{Theorem}[section]
\newtheorem{lem}[thm]{Lemma}

\newtheorem{pro}[thm]{Proposition}

\theoremstyle{definition}
\newtheorem{defn}{Definition}[section]

\begin{document}

\title{Finding Hamiltonian and Longest $(s, t)$-paths of $C$-shaped Supergrid Graphs in Linear Time\footnote{A preliminary version of this paper has appeared in: The International MultiConference of Engineers and Computer Scientists 2019 (IMECS 2019), Hong Kong, vol. I, 2019, pp. 87--93 \cite{Keshavarz19}.}
%\footnote{The research was partly supported by the Ministry of Science and Technology, Taiwan under grant no. MOST 105-2221-E-324-010-MY3.}
}

\author{\vspace{0.5cm}Ruo-Wei Hung$^\textrm{a}$ and Fatemeh Keshavarz-Kohjerdi$^{\mathrm{b},}$\thanks{Corresponding author.}\\
%$^\textrm{b}$ $^,$
%\thanks{Corresponding author.}\\
$^\textrm{a}$\textit{Department of Computer Science \& Information Engineering,}\\
\textit{Chaoyang University of Technology, Wufeng, Taichung 41349, Taiwan}\\
\textit{\vspace{0.2cm}e-mail address: rwhung@cyut.edu.tw}\\
$^\textrm{b}$\textit{Department of Mathematics \& Computer Science,}\\
\textit{Shahed University, Tehran, Iran}\\
\textit{e-mail address: f.keshavarz@shahed.ac.ir}}
%%%%%%%%

%\thanks{$^*$Corresponding author: Ruo-Wei Hung\\
%Department of Computer Science and Information Engineering\\
%Chaoyang University of Technology\\
%Wufeng, Taichung 41349, Taiwan\\
%Tel.: +886-4-23323000 ext. 7758\\
%Fax: +886-4-23742375\\
%E-mail:
%rwhung@cyut.edu.tw}

%%% ----------------------------------------------------------------------
\maketitle
%\thispagestyle{empty}
%%% ----------------------------------------------------------------------

\begin{abstract}
A supergrid graph is a finite vertex-induced subgraph of the infinite graph whose vertex set consists of all points of the plane with integer coordinates and in which two vertices are adjacent if the difference of their $x$ or $y$ coordinates is not larger than $1$. The Hamiltonian path (cycle) problem is to determine whether a graph contains a simple path (cycle) in which each vertex of the graph appears exactly once. This problem is NP-complete for general graphs and it is also NP-complete for general supergrid graphs. Despite the many applications of the problem, it is still open for many classes, including solid supergrid graphs and supergrid graphs with some holes. A graph is called Hamiltonian connected if it contains a Hamiltonian path between any two distinct vertices. In this paper, first we will study the Hamiltonian cycle property of $C$-shaped supergrid graphs, which are a special case of rectangular supergrid graphs with a rectangular hole. Next, we will show that $C$-shaped supergrid graphs are Hamiltonian connected except few conditions. Finally, we will compute a longest path between two distinct vertices in these graphs. The Hamiltonian connectivity of $C$-shaped supergrid graphs can be applied to compute the optimal stitching trace of computer embroidery machines, and construct the minimum printing trace of 3D printers with a $C$-like component being printed.

\vspace{0.2cm}\noindent\textbf{Keywords:}
Hamiltonicity, Hamiltonian connectivity, longest $(s, t)$-path, Supergrid graphs, $C$-shaped supergrid graphs, Computer embroidery machines, 3D printers
\end{abstract}

%====================================================================
\section{Introduction}\label{Introduction}
%====================================================================
A \textit{Hamiltonian path} (\textit{cycle}) in a graph is a simple path (cycle) in which each vertex of the graph appears exactly once. The \textit{Hamiltonian path (cycle) problem} involves deciding whether or not a graph contains a Hamiltonian path (cycle). A graph is called \textit{Hamiltonian} if it contains a Hamiltonian cycle. A graph $G$ is said to be \textit{Hamiltonian connected} if for each pair of distinct vertices $u$ and $v$ of $G$, there is a Hamiltonian path from $u$ to $v$ in $G$. The Hamiltonian path and cycle problems have numerous applications in different areas, including establishing transport routes, production launching, the on-line optimization of flexible manufacturing systems \cite{Ascheuer96}, computing the perceptual boundaries of dot patterns \cite{OCallaghan74}, pattern recognition \cite{Bermond78, PreperataS85, Toussaint80}, DNA physical mapping \cite{Grebinski98}, fault-tolerant routing for 3D network-on-chip architectures \cite{Ebrahimi13}, and so on. It is well known that the Hamiltonian path and cycle problems are NP-complete for general graphs \cite{GareyJ79, Johnson85}. The same holds true for bipartite graphs \cite{Krishnamoorthy76}, split graphs \cite{Golumbic80}, circle graphs \cite{Damaschke89}, undirected path graphs \cite{BertossiB86}, grid graphs \cite{Itai82}, triangular grid graphs \cite{Gordon08}, supergrid graphs \cite{Hung15}, etc.

In the literature, there are many studies for the Hamiltonian connectivity of interconnection networks, including WK-recursive network \cite{Fu08}, recursive dual-net \cite{Li09}, hypercomplete network \cite{Chen00}, alternating group graph \cite{Jwo93}, arrangement graph \cite{Lo01}. The popular hypercubes are Hamiltonian but are not Hamiltonian connected. However, many variants of hypercubes, including augmented hypercubes \cite{Hung12}, generalized base-$b$ hypercube \cite{Huang08}, hypercube-like networks \cite{Park04}, twisted cubes \cite{Huang02}, crossed cubes \cite{Huang00}, M\"{o}bius cubes \cite{Chen04}, folded hypercubes \cite{Hsieh07}, and enhanced hypercubes \cite{Liu11}, have been known to be Hamiltonian connected.

The \textit{longest path problem}, i.e. the problem of finding a simple path with the maximum number of vertices, is one of the most important problems in graph theory. The Hamiltonian path problem is clearly a special case of the longest path problem. Despite the many applications of the problem, it is still open for some classes of graphs, including solid supergrid graphs and supergrid graphs with some holes \cite{Hung16, Hung17a}. There are few classes of graphs in which the longest path problem is polynomial solvable \cite{Bulterman02, Ioannidou11, Keshavarz12b, Mertzios12, Uehara07}. In the area of approximation algorithms, it has been shown that the problem is not in APX, i.e. there is no polynomial-time constant factor approximation algorithm for the problem unless P=NP \cite{Gutin93}. Also,it has been shown that finding a path of length $n-n^\varepsilon$ is not possible in polynomial time unless P=NP \cite{Karger97}. That is, the longest path problem is a very difficult graph problem. In this paper, we focus on supergrid graphs. We will give the necessary and sufficient conditions for the Hamiltonian and Hamiltonian connected of $C$-shaped supergrid graphs. We also present a linear-time algorithm for finding a longest path between any two distinct vertices in a $C$-shaped supergrid graph.

The \emph{two-dimensional integer grid graph} $G^\infty$ is an infinite graph whose vertex set consists of all points of the Euclidean plane with integer coordinates and in which two vertices are adjacent if the (Euclidean) distance between them is equal to 1. The \emph{two-dimensional triangular grid graph} $T^\infty$ is an infinite graph obtained from $G^\infty$ by adding all edges on the lines traced from up-left to down-right. A \textit{grid graph} is a finite vertex-induced subgraph of $G^\infty$ (see Fig. \ref{Fig_ExampleOfGridRelated}(a)). A \emph{triangular grid graph} is a finite vertex-induced subgraph of $T^\infty$ (see Fig. \ref{Fig_ExampleOfGridRelated}(b)). Hung \textit{et al.} \cite{Hung15} have introduced a new class of graphs, namely \textit{supergrid graphs}. The \emph{two-dimensional supergrid graph} $S^\infty$ is an infinite graph obtained from $T^\infty$ by adding all edges on the lines traced from up-right to down-left. A \emph{supergrid graph} is a finite vertex-induced subgraph of $S^\infty$ (see Fig. \ref{Fig_ExampleOfGridRelated}(c)). A solid supergrid graph is a supergrid graph without holes. A \textit{rectangular supergrid graph} is a supergrid graph bounded by a axis-parallel rectangle (see \ref{Fig_ExampleOfRLC}(a)). A \textit{$L$-shaped} or \textit{$C$-shaped} supergrid graph is a supergrid graph obtained from a rectangular supergrid graph by removing a rectangular supergrid graph from it to make a $L$-like or $C$-like shape (see \ref{Fig_ExampleOfRLC}(b) and \ref{Fig_ExampleOfRLC}(c)). The Hamiltonian connectivity and longest $(s, t)$-path of shaped supergrid graphs can be applied in computing the optimal stitching trace of computer embroidery machines \cite{Hung15, Hung17a, Hung17c}. 

\begin{figure}[!t]
\begin{center}
\includegraphics[width=0.8\textwidth]{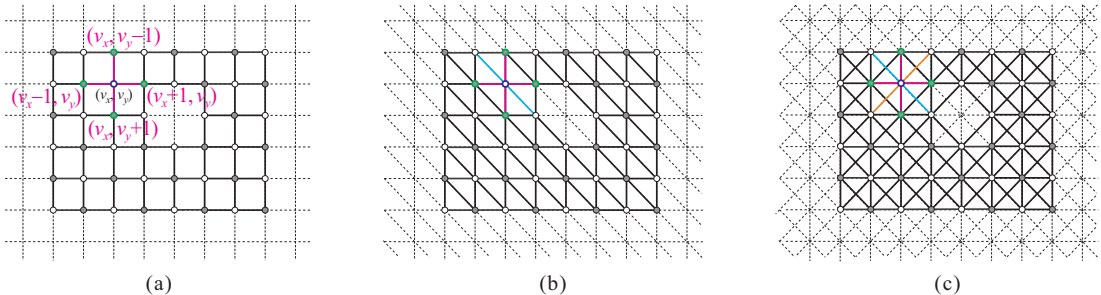}
\caption{(a) A grid graph, (b) a triangular grid graph, and (c) a supergrid graph, where circles represent the vertices and solid lines indicate the edges in the graphs.} \label{Fig_ExampleOfGridRelated}
\end{center}
\end{figure}

\begin{figure}[!t]
\begin{center}
\includegraphics[width=0.7\textwidth]{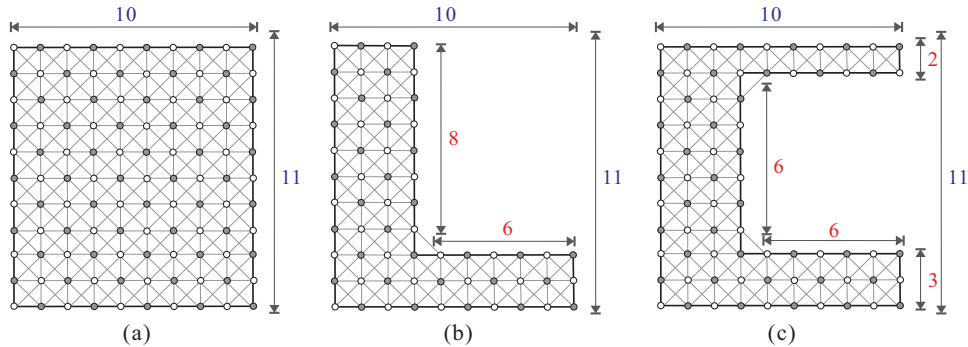}
\caption{(a) A rectangular supergrid graph, (b) a $L$-shaped supergrid graph, and (c) a $C$-shaped supergrid graph, where circles represent the vertices and solid lines indicate the edges in the graphs.} \label{Fig_ExampleOfRLC}
\end{center}
\end{figure}

Previous related works are summarized as follows. Recently, Hamiltonian path (cycle) and Hamiltonian connected problems in grid, triangular grid, and supergrid graphs have received much attention. In \cite{Itai82}, Itai \textit{et al.} proved that the Hamiltonian path problem on grid graphs is NP-complete. They also gave necessary and sufficient conditions for a rectangular grid graph having a Hamiltonian path between two given vertices. Note that rectangular grid graphs are not Hamiltonian connected. Zamfirescu \textit{et al.} \cite{Zamfirescu92} gave sufficient conditions for a grid graph having a Hamiltonian cycle, and proved that all grid graphs of positive width have Hamiltonian line graphs. Later, Chen \textit{et al.} \cite{Chen02} improved the Hamiltonian path algorithm of \cite{Itai82} on rectangular grid graphs and presented a parallel algorithm for the Hamiltonian path problem with two given endpoints in rectangular grid graphs. Also there is a polynomial-time algorithm for finding Hamiltonian cycles in solid grid graphs \cite{Lenhart97}. In \cite{Salman05}, Salman introduced alphabet grid graphs and determined classes of alphabet grid graphs which contain Hamiltonian cycles. Keshavarz-Kohjerdi and Bagheri gave necessary and sufficient conditions for the existence of Hamiltonian paths in alphabet grid graphs, and presented linear-time algorithms for finding Hamiltonian paths with two given endpoints in these graphs \cite{Keshavarz12a}. They also presented a linear-time algorithm for computing the longest path between two given vertices in rectangular grid graphs \cite{Keshavarz12b}, gave a parallel algorithm to solve the longest path problem in rectangular grid graphs \cite{Keshavarz13}, and solved the Hamiltonian path and longest path problems in some classes of grid graphs \cite{Keshavarz16, Keshavarz18a, Keshavarz18b, Keshavarz18c}. Reay and Zamfirescu \cite{Reay00} proved that all 2-connected, linear-convex triangular grid graphs except one special case contain Hamiltonian cycles. The Hamiltonian cycle (path) on triangular grid graphs has been shown to be NP-complete \cite{Gordon08}. They also proved that all connected, locally connected triangular grid graphs (with one exception) contain Hamiltonian cycles.

Recently, Hung \textit{et al.} \cite{Hung15} proved that the Hamiltonian cycle and path problems on supergrid graphs are NP-complete. They also showed that every rectangular supergrid graph always contains a Hamiltonian cycle, proved linear-convex supergrid graphs to be Hamiltonian \cite{Hung16}. Very recently, they verified the Hamiltonian connectivity of rectangular, shaped, alphabet, and $L$-shaped supergrid graphs \cite{Hung17a, Hung17b, Hung17c, Hung18}. They also proposed a linear-time algorithm for the Hamiltonian connected problem on alphabet supergrid graphs \cite{Hung17c}. The Hamiltonian connectivity of $L$-shaped supergrid graphs has been verified in \cite{Hung18, Keshavarz19a}. The $L$-alphabet and $C$-alphabet supergrid graphs in \cite{Hung17c} are special cases of $L$-shaped and $C$-shaped supergrid graphs, respectively. Note that $C$-shaped supergrid graphs contain $L$-shaped supergrid graphs as their subgraphs.

In this paper, we consider the Hamiltonian, the Hamiltonian connectivity, and the longest path of $C$-shaped supergrid graphs, which are special case of rectangular supergrid graph with a rectangular hole. This can be considered as the first attempts to solve the problem in supergrid graphs with some holes.

The rest of the paper is organized as follows. In Section \ref{Sec_Preliminaries}, some notations and observations are given. Previous results are also introduced. In addition, one  special Hamiltonian connected property of $R(m, 3)$, which is a special rectangular supergrid graph, that is used in proving our result will be discovered in this section. In Section \ref{Sec_forbidden-conditions}, we give the necessary and sufficient conditions for the Hamiltonian and Hamiltonian connected of $C$-shaped supergrid graphs. This section shows that $C$-shaped supergrid graphs are always Hamiltonian and Hamiltonian connected except few conditions. In Section \ref{Sec_Algorithm}, we present a linear-time algorithm to compute a longest path between any two distinct vertices in a $C$-shaped supergrid graph. Finally, we make some concluding remarks in Section \ref{Sec_Conclusion}.

%====================================================================
\section{Terminologies and background results}\label{Sec_Preliminaries}
%====================================================================
In this section, we will introduce some terminologies and symbols. Some observations and previously established results for the Hamiltonicity and Hamiltonian connectivity of rectangular and $L$-shaped supergrid graphs are also presented. In addition, we also prove some Hamiltonian connected property of $R(m, 3)$, which is a special rectangular supergrid graph, that will be used in proving our result. For graph-theoretic terminology not defined in this paper, the reader is referred to \cite{Bondy76}.

The \emph{two-dimensional integer grid graph} $G^\infty$ is an infinite graph whose vertex set consists of all points of the Euclidean plane with integer coordinates and in which two vertices are adjacent if the (Euclidean) distance between them is equal to $1$. A \textit{grid graph} is a finite vertex-induced subgraph of $G^\infty$. For a node $v$ in the plane with integer coordinates, let $v_x$ and $v_y$ represent the $x$ and $y$ \textit{coordinates} of node $v$, respectively, denoted by $v=(v_x, v_y)$. If $v$ is a vertex in a grid graph, then its possible adjacent vertices include $(v_x, v_y-1)$, $(v_x-1, v_y)$, $(v_x+1, v_y)$, and $(v_x, v_y+1)$ (see Fig. \ref{Fig_ExampleOfGridRelated}(a)). The \emph{two-dimensional triangular grid graph} $T^\infty$ is an infinite graph obtained from $G^\infty$ by adding all edges on the lines traced from up-left to down-right. A \textit{triangular grid graph} is a finite vertex-induced subgraph of $T^\infty$. If $v$ is a vertex in a triangular grid graph, then its possible neighboring vertices include $(v_x, v_y-1)$, $(v_x-1, v_y)$, $(v_x+1, v_y)$, $(v_x, v_y+1)$, $(v_x-1, v_y-1)$, and $(v_x+1, v_y+1)$ (see Fig. \ref{Fig_ExampleOfGridRelated}(b)). Thus, triangular grid graphs contain grid graphs as subgraphs. The triangular grid graphs defined above are isomorphic to the original triangular grid graphs in \cite{Gordon08} but these graphs are different when considered as geometric graphs.

The \emph{two-dimensional supergrid graph} $S^\infty$ is the infinite graph whose vertex set consists of all points of the plane with integer coordinates and in which two vertices are adjacent if the difference of their $x$ or $y$ coordinates is not larger than $1$. A \textit{supergrid graph} is a finite vertex-induced subgraph of $S^\infty$. The possible adjacent vertices of a vertex $v=(v_x, v_y)$ in a supergrid graph hence include $(v_x, v_y-1)$, $(v_x-1, v_y)$, $(v_x+1, v_y)$, $(v_x, v_y+1)$, $(v_x-1, v_y-1)$, $(v_x+1, v_y+1)$, $(v_x+1, v_y-1)$, and $(v_x-1, v_y+1)$ (see Fig. \ref{Fig_ExampleOfGridRelated}(c)). Thus, supergrid graphs contain grid graphs and triangular grid graphs as subgraphs. Notice that grid and triangular grid graphs are not subclasses of supergrid graphs, and the converse is also true: these classes of graphs have common elements (points) but in general they are distinct since the edge sets of these graphs are different. It is clear that, all grid graphs are bipartite \cite{Itai82} but triangular grid graphs and supergrid graphs are not bipartite. For a vertex $v=(v_x, v_y)$ in a supergrid graph, we color vertex $v$ to be \textit{white} if $v_x+v_y\equiv 0$ (mod 2); otherwise, $v$ is colored to be \textit{black}. Then there are eight possible neighbors of vertex $v$ including four white vertices and four black vertices.

A \textit{rectangular supergrid graph}, denoted by $R(m,n)$, is a supergrid graph whose vertex set is $V(R(m, n))=\{v =(v_x, v_y) | 1\leqslant v_x\leqslant m$ and $1\leqslant v_y\leqslant n\}$. That is, $R(m, n)$ contains $m$ columns and $n$ rows of vertices in $S^\infty$. The size of $R(m, n)$ is defined to be $mn$, and $R(m, n)$ is called $n$-rectangle. $R(m, n)$ is called \textit{even-sized} if $mn$ is even, and it is called \textit{odd-sized} otherwise. Let $v=(v_x, v_y)$ be a vertex in $R(m, n)$. The vertex $v$ is called the \textit{upper-left} (resp., \textit{upper-right}, \textit{down-left}, \textit{down-right}) \textit{corner} of $R(m, n)$ if for any vertex $w=(w_x, w_y)\in R(m, n)$, $w_x\geqslant v_x$ and $w_y\geqslant v_y$ (resp., $w_x\leqslant v_x$ and $w_y\geqslant v_y$, $w_x\geqslant v_x$ and $w_y\leqslant v_y$, $w_x\leqslant v_x$ and $w_y\leqslant v_y$). The edge $(u, v)$ is said to be \textit{horizontal} (resp., \textit{vertical}) if $u_y=v_y$ (resp., $u_x=v_x$), and is called \textit{crossed} if it is neither a horizontal nor a vertical edge. There are four boundaries in a rectangular supergrid graph $R(m, n)$ with $m, n\geqslant 2$. The edge in the boundary of $R(m, n)$ is called \textit{boundary edge}. A path is called \textit{boundary} of $R(m, n)$ if it visits all vertices and edges of the same boundary in $R(m, n)$ and its length equals to the number of vertices in the visited boundary. For example, Fig. \ref{Fig_Rectangular} shows a rectangular supergrid graph $R(10, 8)$ which is called 8-rectangle and contains $2\times(9+7)=32$ boundary edges. Fig. \ref{Fig_Rectangular} also indicates the types of edges and corners. In the figures we will assume that $(1, 1)$ are coordinates of the upper-left corner in a rectangular supergrid graph $R(m, n)$, except we explicitly change this assumption.

\begin{figure}[!t]
\begin{center}
\includegraphics[scale=0.9]{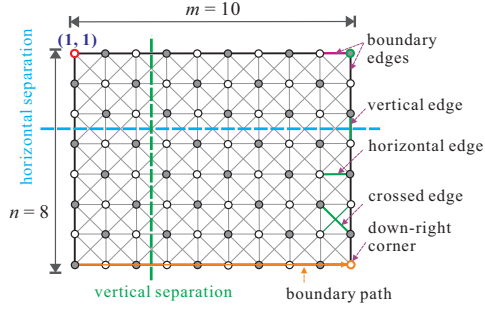}
\caption{A rectangular supergrid graph $R(m, n)$, where $m=10$, $n=8$, and the bold dashed lines indicate vertical and horizontal separations.}\label{Fig_Rectangular}
\end{center}
\end{figure}

A \textit{$L$-shaped supergrid graph}, denoted by $L(m, n; k, l)$, is a supergrid graph obtained from a rectangular supergrid graph $R(m, n)$ by removing its subgraph $R(k, l)$ from the upper-right corner, where $m, n > 1$ and $k, l\geqslant 1$. Then, $m-k\geqslant 1$ and $n-l\geqslant 1$. A \textit{$C$-shaped supergrid graph} $C(m, n; k, l; c, d)$ is a supergrid graph obtained from a rectangular supergrid graph $R(m, n)$ by removing its subgraph $R(k, l)$ from its node coordinated as $(m, c+1)$ while $R(m, n)$ and $R(k, l)$ have exactly one border side in common, where $m\geqslant 2$, $n\geqslant 3$, $k, l\geqslant 1$, $c\geqslant 1$, $d=n-l-c\geqslant 1$, and $a=m-k\geqslant 1$. The structures of $L(m, n; k, l)$ and $C(m, n; k, l; c, d)$ are explained in Fig. \ref{Fig_LC-shaped}(a) and Fig. \ref{Fig_LC-shaped}(b), respectively.

\begin{figure}[!t]
\begin{center}
\includegraphics[width=0.45\textwidth]{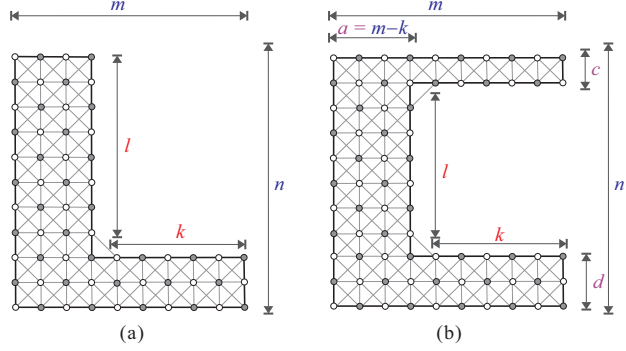}
\caption{The structure of (a) $L$-shaped supergrid graph $L(m, n; k, l)$, where $k=6$, $l=8$, $m-k=4$, and $n-l=3$ and (b) $C$-shaped supergrid graph $C(m, n; k, l; c, d)$, where $k=l=6$, $c=2$, $d=n-l-c=3$, and $a=m-k=4$.}\label{Fig_LC-shaped}
\end{center}
\end{figure}

Let $G = (V, E)$ be a supergrid graph with vertex set $V(G)$ and edge set $E(G)$. Let $S$ be a subset of vertices in $G$, and let $u$ and $v$ be two vertices in $G$. We write $G[S]$ for the subgraph of $G$ \textit{induced} by $S$, $G-S$ for the subgraph $G[V-S]$, i.e., the subgraph induced by $V-S$. In general, we write $G-v$ instead of $G-\{v\}$.
We say that $u$ is \textit{adjacent} to $v$, and $u$ and $v$ are \textit{incident} to edge $(u, v)$, if $(u, v)\in E(G)$. The notation $u\thicksim v$ (resp., $u \nsim v$) means that vertices $u$ and $v$ are adjacent (resp., non-adjacent). A vertex $w$ \textit{adjoins} edge $(u, v)$ if $w\thicksim u$ and $w\thicksim v$. For two edges $e_1=(u_1, v_1)$ and $e_2=(u_2, v_2)$, if $u_1\thicksim u_2$ and $v_1\thicksim v_2$, then we say that $e_1$ and $e_2$ are \textit{parallel}, denoted by $e_1\thickapprox e_2$. For any $v\in V(G)$, a \textit{neighbor} of $v$ is any vertex that is adjacent to $v$. Let $N_G(v)$ be the set of neighbors of $v$ in $G$, and let $N_G[v]=N_G(v)\cup\{v\}$. The \textit{degree} of vertex $v$ in $G$, denoted by $deg(v)$, is the number of vertices adjacent to $v$. A path $P$ of length $|P|$ in $G$, denoted by $v_1\rightarrow v_2\rightarrow \cdots \rightarrow v_{|P|-1} \rightarrow v_{|P|}$, is a sequence $(v_1, v_2, \cdots, v_{|P|-1}, v_{|P|})$ of vertices such that $(v_i,v_{i+1})\in E(G)$ for $1 \leqslant i < |P|$, and all vertices except $v_1, v_{|P|}$ in it are distinct. The first and last vertices visited by $P$ are denoted by $start(P)$ and $end(P)$, respectively. We will use $v_i \in P$ to denote ``$P$ visits vertex $v_i$" and use $(v_i, v_{i+1}) \in P$ to denote ``$P$ visits edge $(v_i, v_{i+1})$". A path from $v_1$ to $v_k$ is denoted by $(v_1, v_k)$-path. In addition, we use $P$ to refer to the set of vertices visited by path $P$ if it is understood without ambiguity. A cycle is a path $C$ with $|V(C)| \geqslant 4$ and $start(C) = end(C)$. Two paths (or cycles) $P_1$ and $P_2$ of graph $G$ are called \textit{vertex-disjoint} if $V(P_1)\cap V(P_2) = \emptyset$. If $end(P_1)\thicksim start(P_2)$, then two vertex-disjoint paths $P_1$ and $P_2$ can be concatenated into a path, denoted by $P_1 \Rightarrow P_2$.

In proving our results, we need to partition a rectangular or $C$-shaped supergrid graph into $\kappa$ disjoint parts, where $\kappa\geqslant 2$. The partition is defined as follows.

\begin{defn}
Let $S$ be a $C$-shaped supergrid graph $C(m, n; k, l; c, d)$ or a rectangular supergrid graph $R(m, n)$. A \textit{separation operation} on $S$ is a partition of $S$ into $\kappa$ vertex-disjoint rectangular supergrid subgraphs $S_1$, $S_2$, $\cdots$, $S_\kappa$, i.e., $V(S)=V(S_1)\cup V(S_2)\cup \cdots \cup V(S_\kappa)$ and $V(S_i)\cap V(S_j) =\emptyset$ for $i\neq j$ and $1\leqslant i, j\leqslant \kappa$, where $\kappa\geqslant 2$. A separation is called \textit{vertical} if it consists of a set of horizontal edges, and is called \textit{horizontal} if it contains a set of vertical edges. For an example, the bold dashed vertical (resp., horizontal) line in Fig. \ref{Fig_Rectangular} indicates a vertical (resp., horizontal) separation of $R(10, 8)$ which partitions it into $R(3, 8)$ and $R(7, 8)$ (resp., $R(10, 3)$ and $R(10, 5)$).
\end{defn}

Let $R(m, n)$ be a rectangular supergrid graph with $m\geqslant n\geqslant 2$, $\mathcal{C}$ be a cycle of $R(m, n)$, and let $H$ be a boundary of $R(m, n)$, where $H$ is a subgraph of $R(m, n)$. The restriction of $\mathcal{C}$ to $H$ is denoted by $\mathcal{C}_{| H}$. If $|\mathcal{C}_{| H}|=1$, i.e. $\mathcal{C}_{| H}$ is a boundary path on $H$, then $\mathcal{C}_{| H}$ is called \textit{flat face} on $H$. If $|\mathcal{C}_{| H}|>1$ and $\mathcal{C}_{| H}$ contains at least one boundary edge of $H$, then $\mathcal{C}_{| H}$ is called \textit{concave face} on $H$. A Hamiltonian cycle of $R(m, 3)$ is called \textit{canonical} if it contains three flat faces on two shorter boundaries and one longer boundary, and it contains one concave face on the other boundary, where the shorter boundary consists of three vertices. And, a Hamiltonian cycle of $R(m, n)$ with $n=2$ or $n\geqslant 4$ is said to be \textit{canonical} if it contains three flat faces on three boundaries, and it contains one concave face on the other boundary. The following lemma states the result in \cite{Hung15} concerning the Hamiltonicity of rectangular supergrid graphs.

\begin{lem}\label{HC-rectangular_supergrid_graphs}
(See \cite{Hung15}) Let $R(m, n)$ be a rectangular supergrid graph with $m\geqslant n\geqslant 2$. Then, the following statements hold true:\\
$(1)$ if $n=3$, then $R(m, 3)$ contains a canonical Hamiltonian cycle;\\
$(2)$ if $n=2$ or $n\geqslant 4$, then $R(m, n)$ contains four canonical Hamiltonian cycles with concave faces being on different boundaries.
\end{lem}

\begin{figure}[!t]
\begin{center}
\includegraphics[width=0.95\textwidth]{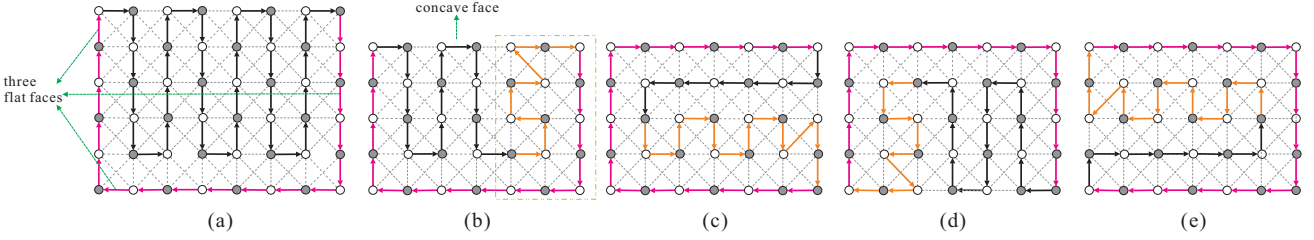}
\caption{A canonical Hamiltonian cycle containing three flat faces and one concave face for (a) $R(8, 6)$ and (b)--(e) $R(7, 5)$, where solid arrow lines indicate the edges in the cycles and $R(7, 5)$ contains four distinct canonical Hamiltonian cycles in (b)--(e) such that their concave faces are placed on different boundaries.} \label{Fig_HC-Rectangular}
\end{center}
\end{figure}

Fig. \ref{Fig_HC-Rectangular} shows canonical Hamiltonian cycles for even-sized and odd-sized rectangular supergrid graphs found in Lemma \ref{HC-rectangular_supergrid_graphs}. Each Hamiltonian cycle found by this lemma contains all the boundary edges on any three sides of the rectangular supergrid graph. This shows that for any rectangular supergrid graph  $R(m, n)$ with $m\geqslant n\geqslant 4$, we can always construct four canonical Hamiltonian cycles such that their concave faces are placed on different boundaries. For instance, the four distinct canonical Hamiltonian cycles of $R(7, 5)$ are shown in Fig. \ref{Fig_HC-Rectangular}(b)--(e), where the concave faces of these four canonical Hamiltonian cycles are located on different boundaries.

Let $(G, s, t)$ denote the supergrid graph $G$ with two specified distinct vertices $s$ and $t$. Without loss of generality, we will assume that $s_x \leqslant t_x$ in the rest of the paper. We denote a Hamiltonian path between $s$ and $t$ in $G$ by $HP(G, s, t)$. We say that $HP(G, s, t)$ does exist if there is a Hamiltonian $(s, t)$-path in $G$. From Lemma \ref{HC-rectangular_supergrid_graphs}, we know that $HP(R(m, n), s, t)$ does exist if $m, n\geqslant 2$ and $(s, t)$ is an edge in the constructed Hamiltonian cycle of $R(m, n)$.

\begin{defn}
Assume that $G$ is a connected supergrid graph and $V_1$ is a subset of the vertex set $V(G)$. $V_1$ is a \textit{vertex cut} if $G-V_1$ is disconnected. A vertex $v\in V(G)$ is a \textit{cut vertex}, if $G-\{v\}$ is disconnected. For an example, in Fig. \ref{Fig_ForbiddenConditionF1}(b) $\{s, t\}$ is a vertex cut, and in Fig. \ref{Fig_ForbiddenConditionF1}(a) $t$ is a cut vertex.
\end{defn}

\begin{figure}[!t]
\begin{center}
\includegraphics[scale=0.9]{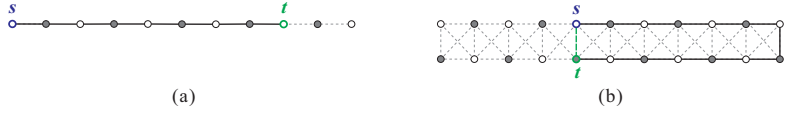}
\caption{Rectangular supergrid graphs in which there is no Hamiltonian $(s, t)$-path for (a) $R(m, 1)$, and (b) $R(m, 2)$, where solid lines indicate the longest path between $s$ and $t$.} \label{Fig_ForbiddenConditionF1}
\end{center}
\end{figure}

In \cite{Hung17a}, the authors showed that $HP(R(m,n), s, t)$ does not exist if the following condition hold:
%The following condition then implies that $HP(R(m, 1), s, t)$ and $HP(R(m, 2), s, t)$ do not exist.\\

\begin{verse}
(F1) $s$ or $t$ is a cut vertex, or $\{s,t\}$ is a vertex cut (see Fig. \ref{Fig_ForbiddenConditionF1}(a) and Fig. \ref{Fig_ForbiddenConditionF1}(b)).
\end{verse}

Let $G$ be any supergrid graphs. The following lemma showing that $HP(G, s, t)$ does not exist if $(G, s, t)$ satisfies condition (F1) can be verified by the arguments in \cite{Keshavarz16}.

\begin{lem}\label{Lemma:1b}(See \cite{Keshavarz16})
Let $G$ be a supergrid graph with two vertices $s$ and $t$. If $(G, s, t)$ satisfies condition $\mathrm{(F1)}$, then $HP(G, s, t)$ does not exist.
\end{lem}

The Hamiltonian $(s, t)$-path $P$ of $R(m, n)$ constructed in \cite{Hung17a} satisfies that $P$ contains at least one boundary edge of each boundary, and is called \textit{canonical}.

\begin{lem}\label{HamiltonianConnected-Rectangular}
(See \cite{Hung17a}) Let $R(m, n)$ be a rectangular supergrid graph with $m, n \geqslant 1$, and let $s$ and $t$ be its two distinct vertices. If $(R(m, n), s, t)$ does not satisfy condition $\mathrm{(F1)}$, then there exists a canonical Hamiltonian $(s, t)$-path of $R(m, n)$, i.e., $HP(R(m, n), s, t)$ does exist.
\end{lem}

Consider that $(R(m, n), s, t)$ does not satisfy condition (F1). Let $w=(1, 1)$, $z=(2, 1)$, and $f=(3, 1)$ be three vertices of $R(m, n)$ with $m\geqslant 3$ and $n\geqslant 2$. In \cite{Keshavarz19a}, we have proved that there exists a Hamiltonian $(s, t)$-path $Q$ of $R(m, n)$ such that $(z, f)\in Q$ if the following condition (F2) holds; and $(w, z)\in Q$ otherwise.

\begin{verse}
(F2) $n=2$ and $\{s, t\}\in \{\{w, z\}, \{(1, 1), (2, 2)\}, \{(2, 1), (1, 2)\}\}$, or $n\geqslant 3$ and $\{s, t\}=\{w, z\}$.
\end{verse}

The above result is presented as follows.

\begin{lem}\label{HamiltonianConnected-Rectangular-wz_rectangle}
(See \cite{Keshavarz19a}) Let $R(m, n)$ be a rectangular supergrid graph with $m\geqslant 3$ and $n\geqslant 2$, $s$ and $t$ be its two distinct vertices, and let $w=(1, 1)$ and $z=(2, 1)$. If $(R(m, n), s, t)$ does not satisfy condition $\mathrm{(F1)}$, then there exists a canonical Hamiltonian $(s, t)$-path $Q$ of $R(m, n)$ such that $(z, f)\in Q$ if $(R(m, n), s, t)$ does satisfy condition $\mathrm{(F2)}$; and $(w, z)\in Q$ otherwise.
\end{lem}

We then give some observations on the relations among cycle, path, and vertex. These propositions will be used in proving our results and are given in \cite{Hung15, Hung16, Hung17a}.

\begin{pro}\label{Pro_Obs}
(See \cite{Hung15, Hung16, Hung17a}) Let $C_1$ and $C_2$ be two vertex-disjoint cycles of a graph $G$, let $C_1$ and $P_1$ be a cycle and a path, respectively, of $G$ with $V(C_1)\cap V(P_1)=\emptyset$, and let $x$ be a vertex in $G-V(C_1)$ or $G-V(P_1)$. Then, the following statements hold true:\\
$(1)$ If there exist two edges $e_1\in C_1$ and $e_2\in C_2$ such that $e_1 \thickapprox e_2$, then $C_1$ and $C_2$ can be combined into a cycle of $G$ (see Fig. \emph{\ref{Fig_Obs}(a)}).\\
$(2)$ If there exist two edges $e_1\in C_1$ and $e_2\in P_1$ such that $e_1 \thickapprox e_2$, then $C_1$ and $P_1$ can be combined into a path of $G$ (see Fig. \emph{\ref{Fig_Obs}(b)}). \\
$(3)$ If vertex $x$ adjoins one edge $(u_1, v_1)$ of $C_1$ (resp., $P_1$), then $C_1$ (resp., $P_1$) and $x$ can be combined into a cycle (resp., path) of $G$ (see Fig. \emph{\ref{Fig_Obs}(c)}).\\
$(4)$ If there exists one edge $(u_1, v_1)\in C_1$ such that $u_1\thicksim start(P_1)$ and $v_1\thicksim end(P_1)$, then $C_1$ and $P_1$ can be combined into a cycle $C$ of $G$ (see Fig. \emph{\ref{Fig_Obs}(d)}).
\end{pro}

\begin{figure}[!t]
\begin{center}
\includegraphics[width=0.7\textwidth]{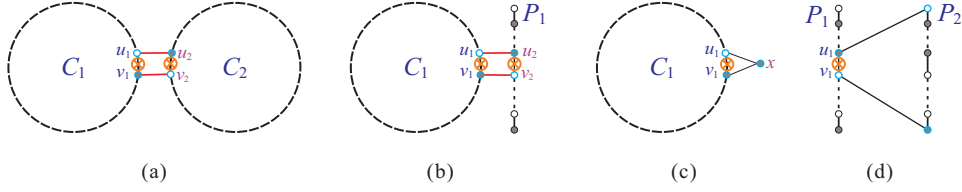}
\caption{A schematic diagram for (a) Statement (1), (b) Statement (2), (c) Statement (3), and (d) Statement (4) of Proposition \ref{Pro_Obs}, where bold dashed lines indicate the cycles (paths) and $\otimes$ represents the destruction of an edge while constructing a cycle or path.} \label{Fig_Obs}
\end{center}
\end{figure}

Next, we will discover one Hamiltonian connected property of 3-rectangle $R(m, 3)$ with $m\geqslant 3$ that will be used in proving our result. Let $z_1=(m, 1)$, $z_2=(m, 2)$, and $z_3=(m, 3)$ be three vertices of $R(m, 3)$. Let $\tilde{R}=R(m, 3)-\{z_1, z_2, z_3\}$ and edges $e_{12}=(z_1, z_2)$, $e_{23}=(z_2, z_3)$. Then, $\tilde{R}=R(m-1, 3)$. Let $s, t\in \tilde{R}$. We will prove that there exists a Hamiltonian $(s, t)$-path $P$ of $R(m, 3)$ such that $e_{12}, e_{23}\in P$. Before giving this property, we first give one result in \cite{Hung17a} for 3-rectangle as follows.

\begin{lem}\label{HP-3rectangle}
(See \cite{Hung17a}) Let $R(m, 3)$ be a $3$-rectangle with $m\geqslant 3$, and let $s$ and $t$ be its two distinct vertices. Then, $R(m, 3)$ contains a canonical Hamiltonian $(s, t)$-path $P$ which contains at least one boundary edge of each boundary in $R(m, 3)$.
\end{lem}

By using the above lemma, we will prove the following lemma.

\begin{lem}\label{HP-3rectangle-boundary_path}
Let $R(m, 3)$ be a $3$-rectangle with $m\geqslant 3$, and let $s$ and $t$ be its two distinct vertices. Let $z_1=(m, 1)$, $z_2=(m, 2)$, and $z_3=(m, 3)$ be three vertices of $R(m, 3)$, and let edges $e_{12}=(z_1, z_2)$, $e_{23}=(z_2, z_3)$. If $\{s, t\}\cap\{z_1, z_2, z_3\}=\emptyset$, then there exists a Hamiltonian $(s, t)$-path of $R(m, 3)$ containing $e_{12}$ and $e_{23}$.
\end{lem}
\begin{proof}
We will prove this lemma by induction on $m$. Let $\tilde{R}=R(m, 3)-\{z_1, z_2, z_3\}$. Then, $\tilde{R}=R(m-1, 3)$, where $m-1\geqslant 2$. Initially, let $m=3$. Then, $\tilde{R}=R(2, 3)$ and $s, t\in \tilde{R}$. By considering every case, we can construct the desired Hamiltonian $(s, t)$-path of $R(3, 3)$, as shown in Fig. \ref{Fig_HP-3rectangle-boundary_path}(a)--(o). 
%Note that $s_x$ may be larger than $t_x$ in the figures.
 Assume that the lemma holds true when $m=k\geqslant 3$. Consider that $m=k+1$. Then, $\tilde{R}=R(k, 3)$ is a subgraph of $R(k+1, 3)$, where $z_1=(k+1, 1)$, $z_2=(k+1, 2)$, $z_3=(k+1, 3)$, and $s, t\in \tilde{R}=R(k, 3)$. Let $\hat{P} = z_1\rightarrow z_2\rightarrow z_3$. By Lemma \ref{HP-3rectangle}, $\tilde{R}$ contains a Hamiltonian $(s, t)$-path $\tilde{P}$ such that it contains an edge $\tilde{e}=(u, v)$ locating to face $R(k+1, 3)-\tilde{R}$. Then, $start(\hat{P})\thicksim u$ and $end(\hat{P})\thicksim v$. By Statement (4) of Proposition \ref{Pro_Obs}, $\tilde{P}$ and $\hat{P}$ can be combined into a Hamiltonian $(s, t)$-path of $R(k+1, 3)$. The construction of such a Hamiltonian path is depicted in Fig. \ref{Fig_HP-3rectangle-boundary_path}(p). Thus, the lemma holds when $m=k+1$. By induction, the lemma holds true.
\end{proof}

\begin{figure}[!t]
\begin{center}
\includegraphics[width=0.85\textwidth]{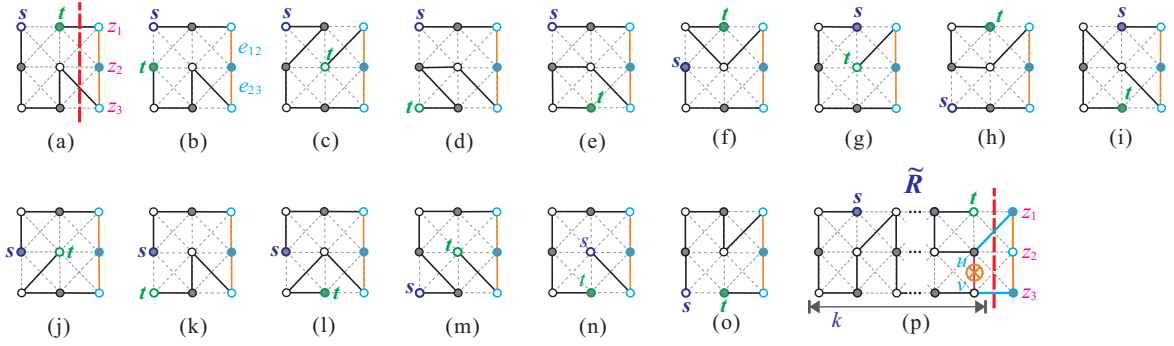}
\caption{(a)--(o) The possible Hamiltonian $(s, t)$-path of $R(3, 3)$ when $s, t\not\in\{z_1, z_2, z_3\}$, and (p) the construction of Hamiltonian $(s, t)$-path of $R(k+1, 3)$ for $k\geqslant 3$ and $s, t\in \tilde{R}$, where the solid lines indicate the constructed Hamiltonian $(s, t)$-path and $\otimes$ represents the destruction of an edge while constructing a Hamiltonian $(s, t)$-path of $R(k+1, 3)$.} \label{Fig_HP-3rectangle-boundary_path}
\end{center}
\end{figure}

In addition to condition (F1) (as depicted in Fig. \ref{Fig_ForbiddenConditionF2F5}(a) and  \ref{Fig_ForbiddenConditionF2F5}(b)), in \cite{Keshavarz19a}, we showed that $HP(L(m,n; k, l), s, t)$ does not exist whenever one of the following conditions is satisfied.

\begin{verse}
(F3) assume that $G$ is a supergrid graph, there exists a vertex $w \in G$ such that $deg(w) = 1$, $w \neq s$, and $w \neq t$ (see Fig. \ref{Fig_ForbiddenConditionF2F5}(c)).
\end{verse}

\begin{verse}
(F4) $m - k = 1$, $n - l = 2$, $l = 1$, $k \geq 2$, and $\{s, t\} = \{(1, 2), (2, 3)\}$ or $\{(1, 3), (2, 2)\}$ (see Fig. \ref{Fig_ForbiddenConditionF2F5}(d)).
\end{verse}

\begin{figure}[!t]
\begin{center}
\includegraphics[width=0.8\textwidth]{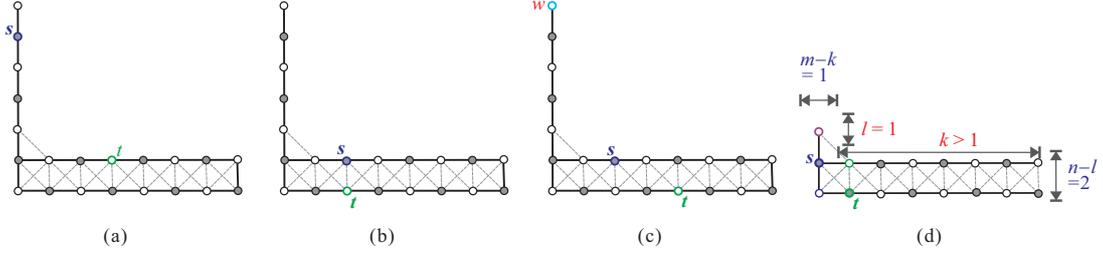}
\caption{$L$-shaped supergrid graph in which there is no Hamiltonian $(s, t)$-path for (a) $s$ is a cut vertex, (b) $\{s, t\}$ is a vertex cut, (c) there exists a vertex $w$ such that $deg(w)=1$, $w\neq s$, and $w\neq t$, and (d) $m-k=1$, $n-l=2$, $l=1$, $k\geqslant 2$, and $\{s, t\}=\{(1, 2), (2, 3)\}$.} \label{Fig_ForbiddenConditionF2F5}
\end{center}
\end{figure}

\begin{thm}\label{HamiltonianConnected-L-shaped}
(See \cite{Keshavarz19a}) Let $L(m, n; k, l)$ be a $L$-shaped supergrid graph with vertices $s$ and $t$. If $(L(m, n; k, l), s, t)$ does not satisfy conditions $\mathrm{(F1)}$, $\mathrm{(F3)}$, and $\mathrm{(F4)}$, then $L(m, n; k, l)$ contains a Hamiltonian $(s, t)$-path, i.e., $HP(L(m, n; k, l), s, t)$ does exist.
\end{thm}

\begin{thm} \label{HC-Lshaped}(See \cite{Hung18})
Let $L(m, n; k, l)$ be a L-shaped supergrid graph. Then, $L(m, n; k, l)$ contains a Hamiltonian cycle if it does not satisfy condition $\mathrm{(F5)}$, where condition $\mathrm{(F5)}$ is defined as follows:
\end{thm}
\begin{verse}
(F5) there exists a vertex $w$ in $L(m, n; k, l)$ such that $deg(w) = 1$.
\end{verse}

In the following, we use $\hat{L}(G,s,t)$ to denote the length of longest paths between $s$ and $t$ and $\hat{U}(G,s,t)$ to indicate the upper bound on the length of longest paths between $s$ and $t$, where $G$ is a rectangular, $L$-shaped, or $C$-shaped supergrid graph. By the length of a path we mean the number of vertices of the path. Let $G$ be a rectangular supergrid graph $R(m,n)$ or $L$-shaped supergrid graphs $L(m,n;k,l)$. In \cite{Hung17a, Keshavarz19a}, the authors proved the following upper bounds on the length of longest paths in $G$:\\

$\hat{U}(G, s, t)=
  \begin{cases}
   t_{x}-s_{x}+1,                               & \mathrm{if} \ G=R(m,n)\ \mathrm{and}\ n=1,\\
   \max\{2s_x, 2(m-s_x+1)\}\ \mathrm{or} \ 2m,  & \mathrm{if} \ G=R(m,n)\ \mathrm{and}\ n=2,\\
   mn,                                          & \mathrm{if }\ G=R(m,n)\ \mathrm{and}\ n\geqslant 3,\\
   |t_y-s_y|+1,                                 & \mathrm{if} \ G=L(m,n;k,l)\ \mathrm{and\ (FC1)\ holds,} \\
   n-s_y+t_x,                                   & \mathrm{if}\ G=L(m,n;k,l)\ \mathrm{and\ (FC2)\ holds,}\\
   |t_y-s_y|+2,                                 & \mathrm{if} \ G=L(m,n;k,l)\ \mathrm{and\ (FC3)\ holds,} \\
   \hat{L}(G', s, t),                           & \mathrm{if}\ G=L(m,n;k,l)\ \mathrm{and\ (FC4),\ (FC5),\ or\ (FC6d)\ holds}, \\
   \max\{2s_y, mn-kl-2s_y+2\},                  & \mathrm{if}\ G=L(m,n;k,l)\ \mathrm{and \ (FC6a)\ holds,}\\
   \max\{2s_y-1, mn-2s_y+1\},                  & \mathrm{if}\ G=L(m,n;k,l)\ \mathrm{and \ (FC6b)\ holds,}\\
   \max\{2(m-s_x+1), mn-kl-2(m-s_x+1)+2\},      & \mathrm{if}\ G=L(m,n;k,l)\ \mathrm{and \ (FC6c)\ holds,}\\
   %max( L(G_1,s,t), L(G_2,s,t)), &         \mathrm{if}\G=L(m,n;k,l)\ \mathrm{and \ (FC6d)\ holds,}\\
   mn-kl-1,                                     & \mathrm{if}\ G=L(m,n;k,l)\ \mathrm{and \ (F4)\ holds,} \\
   mn-kl,                                       & \mathrm{if}\ G=L(m,n;k,l)\ \mathrm{and \ (C0)\ holds,}\\
  \end{cases}$
where $(\mathrm{C0})$, $(\mathrm{FC1})$, $(\mathrm{FC2})$, $(\mathrm{FC3})$, $(\mathrm{FC4})$,  $(\mathrm{FC5})$, $(\mathrm{FC6a})$, $(\mathrm{FC6b})$, $(\mathrm{FC6c})$, $(\mathrm{FC6d})$ are defined as follows:

\begin{verse}
(C0) $(L(m,n; k,l),s,t)$ does not satisfy any of conditions (F1), (F3), and (F4).
\end{verse}

\begin{verse}
(FC1) $m-k = n-l = 1$, $l > 1$, and $s_y, t_y \leqslant l$.
\end{verse}

\begin{verse}
(FC2) $m-k = n-l = 1$, $l > 1$, $s_y < l$, and $t_x > 1$.
\end{verse}
\begin{verse}
(FC3) $m-k = n-l = 1$, $l > 1$, $s_x = t_x = 1$, $\max\{s_y, t_y\} = n$, and $[(k>1)$ or $(k=1$ and $\min\{s_y, t_y\} > 1)]$.
\end{verse}

\begin{verse}
(FC4) $n-l > 1$, $m-k = 1$, $l > 1$, and $[(s_y, t_y > l$ and $\{s, t\}$ is not a vertex cut$)$, $(s_y\leqslant l$ and $t_y > l)$, or $(t_y\leqslant l$ and $s_y > l)]$. Here, $G' = L(m, n-n'; k, l')$, where $l' = l - n'$, and $n' = l-1$ if $s_y, t_y\geqslant l$; otherwise $n' = \min\{s_y, t_y\} - 1$.
\end{verse}

\begin{verse}
(FC5) $n-l > 1$, $m-k = 1$, $m > 2$, $s = (1, l+1)$, and $t = (2, l+1)$. Here, $G' = R(m, n-l)$.
\end{verse}

\begin{verse}
(FC6a) $l > 1$, $n-l > 1$, $m-k = 2$, and $2\leqslant s_y=t_y\leqslant l$.
\end{verse}

\begin{verse}
(FC6b) $m=2$, $k =l = 1$, $n-l > 1$, and $l+1\leqslant s_y=t_y\leqslant n-1$.
\end{verse}

\begin{verse}
(FC6c) $k > 1$, $l = 1$, $m-k = 1$, $n-l = 2$, and $2\leqslant s_x=t_x\leqslant m-1$.
\end{verse}

\begin{verse}
(FC6d) $(m = 2$, $k = 1$, $l > 1$, $n-l > 1$, and $l+1 \leqslant s_y=t_y\leqslant n-1)$ or $(k, l > 1$, $m-k = 1$, $n-l = 2$, and $2\leqslant s_x=t_x\leqslant m-1$). Here, $G' = L(m, n-l+1; k, 1)$.
\end{verse}

\begin{thm}\cite{Keshavarz19a}\label{L-shaped_LongestPath}
Given a rectangular supergrid graph $R(m, n)$ with $mn\geqslant 2$ or $L$-shaped supergrid graph $L(m, n; k, l)$, and two distinct vertices $s$ and $t$ in $R(m, n)$ or $L(m, n; k, l)$, a longest $(s, t)$-path can be found in $O(mn)$-linear time.
\end{thm}

%====================================================================
\section{The necessary and sufficient conditions for the Hamiltonian and Hamiltonian connected of $C$-shaped supergrid graphs}\label{Sec_forbidden-conditions}
%====================================================================
In this section, we will give necessary and sufficient conditions for $C$-shaped supergrid graphs to have a Hamiltonian cycle and Hamiltonian $(s,t)$-path. First, we will verify the Hamiltonicity of $C$-shaped supergrid graphs. If $a(=m-k)=1$ or there exists a vertex $w\in V(C(m, n; k, l;c,d))$ such that $deg(w) = 1$, then $C(m, n; k, l;c,d)$ contains no Hamiltonian cycle. Therefore, $C(m, n; k, l;c,d)$ is not Hamiltonian if condition (F6) is satisfied, where (F6) is defined as follows:

\begin{verse}
(F6) $a(=m-k) = 1$ or there exists a vertex $w \in V(C(m, n; k, l; c, d))$ such that $deg(w) = 1$.
\end{verse}

\begin{figure}[h]
\centering
\includegraphics[scale=0.9]{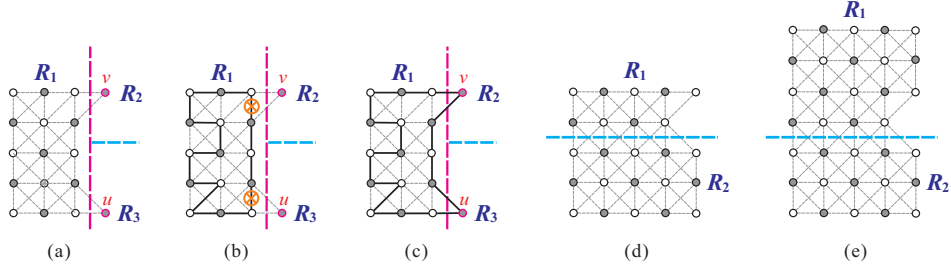}
\caption{(a) A vertical and horizontal separations on $C(m,n; k,l; c,d)$, (b) a Hamiltonian cycle in $R_1$, (c) a Hamiltonian cycle in $C(m,n; k,l; c,d)$, and (d) and (e) a horizontal separation on $C(m,n; k,l; c,d)$, where bold solid lines indicate the constructed Hamiltonian cycle and $\otimes$ represents the destruction of an edge while constructing a Hamiltonian cycle.}\label{fig:HamCycle-Cshaped1}
\end{figure}

\begin{thm}
$C(m,n; k,l; c,d)$ contains a Hamiltonian cycle if and only if it does not satisfy condition $\mathrm{(F6)}$.
\end{thm}
\begin{proof}
Only if part $(\Rightarrow)$: Assume that $C(m,n; k,l; c,d)$ satisfies condition (F6), then we show that it contains no Hamiltonian cycle. Let $v\in V(C(m,n; k,l; c,d))$ such that $c+1\leqslant v_y\leqslant c+l$ if $a(=m-k) = 1$; otherwise $v\thicksim w$. It is obvious that any cycle in $C(m,n; k,l; c,d)$ must pass through $v$ two times. Therefore, $C(m,n; k,l; c,d)$ contains no Hamiltonian cycle.

If part $(\Leftarrow)$: We prove this statement by constructing a Hamiltonian cycle of $C(m,n; k,l; c,d)$. We consider the following two cases:

Case 1: $c=1$ and $d=1$. In this case, $k = 1$. If $k > 1$, then there exists a vertex $w \in V(C(m,n; k,l; c,d))$ such that $deg(w)=1$. We make a vertical and horizontal separations on $C(m,n; k,l; c,d)$ to obtain three disjoint rectangular supergrid subgraphs $R_1 = R(a, n)$, $R_2 = R(k, c)$, and $R_3 = R(k, d)$, as depicted in Fig. \ref{fig:HamCycle-Cshaped1}(a). Assume that $v\in V(R_2)$ and $u\in V(R_3)$. By Lemma \ref{HC-rectangular_supergrid_graphs}, $R_1$ contains a canonical Hamiltonian cycle $HC_1$ (see Fig. \ref{fig:HamCycle-Cshaped1}(b)). We can place one flat face of $HC_1$ to face $R_2$ and $R_3$. Thus, there exists an edge $(w, z)\in HC_1$ such that $v \thicksim w$ and $v \thicksim z$. By Statement (3) of Proposition \ref{Pro_Obs}, $v$ and $HC_1$ can be combined into a cycle $HC_2$. By the same argument, $u$ can be merged into the cycle $HC_2$ to form a Hamiltonian cycle of $C(m,n; k,l; c,d)$, as shown in Fig. \ref{fig:HamCycle-Cshaped1}(c).

Case 2: $c\geqslant 2$ or $d\geqslant 2$. By symmetry, assume that $d\geqslant 2$. We make a horizontal separation on $C(m,n; k,l; c,d)$ to obtain two disjoint supergrid subgraphs $R_1 = L(m, c+l; k,l)$ and $R_2 = R(m, d)$, as depicted in Fig. \ref{fig:HamCycle-Cshaped1}(d) and \ref{fig:HamCycle-Cshaped1}(e), where Fig. \ref{fig:HamCycle-Cshaped1}(d) and Fig. \ref{fig:HamCycle-Cshaped1}(e) respectively indicate the case of $c=1$ and $c\geqslant 2$. By Theorem \ref{HC-Lshaped} (resp. Lemma \ref{HC-rectangular_supergrid_graphs}), $R_1$ (resp. $R_2$) contains a Hamiltonian (resp. canonical Hamiltonian) cycle $HC_1$ (resp. $HC_2$) such that its one flat face is placed to face $R_2$ (resp. $R_1$). Then, there exist two edges $e_1 = (u_1, u_2) \in HC_1$ and $e_2 = (v_1, v_2) \in HC_2$ such that $e_1 \thickapprox e_2$; as shown in Fig. \ref{fig:HamCycle-Cshaped2}(a) and \ref{fig:HamCycle-Cshaped2}(b). By Statement (1) of Proposition \ref{Pro_Obs}, $HC_1$ and $HC_2$ can be combined into a Hamiltonian cycle of $C(m,n; k,l; c,d)$, as shown in Fig. \ref{fig:HamCycle-Cshaped2}(c) and \ref{fig:HamCycle-Cshaped2}(d).
\end{proof}

\begin{figure}[h]
\centering
\includegraphics[scale=0.9]{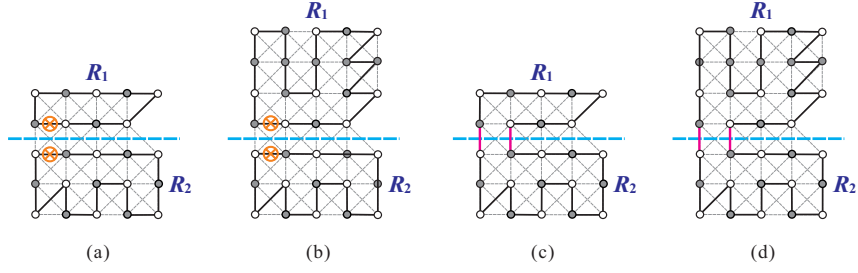}
\caption{(a) and (b) Hamiltonian cycles in $R_1$ and $R_2$, and (c) and (d) a Hamiltonian cycle in $C(m,n; k,l; c,d)$ for Fig. \ref{fig:HamCycle-Cshaped1}(d) and Fig. \ref{fig:HamCycle-Cshaped1}(e) respectively.}\label{fig:HamCycle-Cshaped2}
\end{figure}
%===================================================================

Now, we give necessary and sufficient conditions for the existence of a Hamiltonian $(s, t)$-path in $C(m,n; k,l; c,d)$. In addition to condition (F1) (as depicted in Fig. \ref{fig:FC1}(a)--\ref{fig:FC1}(b)) and (F3) (as depicted in Fig. \ref{fig:FC1}(c)), if $(C(m,n; k,l; c,d), s, t)$ satisfies one of the following conditions, then it contains no Hamiltonian $(s, t)$-path.

\begin{figure}[!t]
\centering
\includegraphics[scale=0.9]{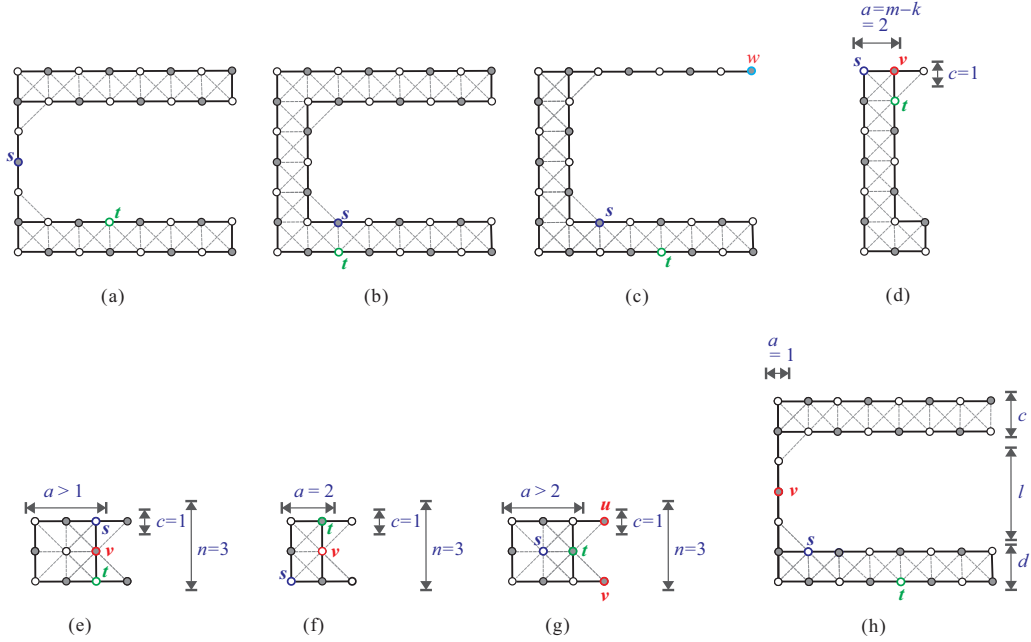}
\caption{Some $C$-shaped supergrid graphs in which there is no Hamiltonian $(s, t)$-path.}\label{fig:FC1}
\end{figure}

\begin{verse}
(F7) $m = 3$, $a (= m-k) = 2$, and $[(c = 1$ and $\{s, t\} = \{(1, 1), (2, 2)\}$ or $\{(1, 2), (2, 1)\})$ or $(d = 1$ and $\{s, t\} = \{(1, n), (2, n-1)\}$ or $\{(1, n - 1), (2, n)\})]$ (see Fig. \ref{fig:FC1}(d)).
\end{verse}

\begin{verse}
(F8) $n = 3$, $k = c = d = 1$, and\\
\hspace{0.6cm}(1) $a \geqslant 2$, $s_x=t_x=m-1$, and $|s_y-t_y|=2$ (see Fig. \ref{fig:FC1}(e)); or\\
\hspace{0.6cm}(2) $a = 2$, $s_x = 1$, $t_x = 2$, and $|s_y-t_y|=2$ (see Fig. \ref{fig:FC1}(f)); or\\
\hspace{0.6cm}(3) $a > 2$, $s_x < m - 1$, and $t = (m-1, 2)$ (see Fig. \ref{fig:FC1}(g)).
\end{verse}

\begin{verse}
(F9) $a( = m - k) = 1$, and ($s_y, t_y\leqslant c$ or $s_y, t_y > c + l$) (see Fig. \ref{fig:FC1}(h)).
\end{verse}

\begin{lem}\label{Necessary-condition-Cshaped}
If $HP(C(m,n; k,l; c,d), s, t)$ exists, then $(C(m,n; k,l; c,d), s, t)$ does not satisfy conditions $\mathrm{(F1)}$, $\mathrm{(F3)}$, $\mathrm{(F7)}$, $\mathrm{(F8)}$, and $\mathrm{(F9)}$.
\end{lem}
\begin{proof}
Assume that $(C(m,n; k,l; c,d), s, t)$ satisfies one of the conditions (F1), (F3), (F7), (F8), and (F9), then we show that $HP(C(m,n; k,l; c,d), s, t)$ does not exist. For conditions (F1) and (F3), it is clear (see Fig. \ref{fig:FC1}(a)--(c)). For condition (F7), by inspecting all cases of Fig. \ref{fig:FC1}(d) there exists no Hamiltonian $(s, t)$-path. For cases (1)--(2) of condition (F8), consider Fig. \ref{fig:FC1}(e) and \ref{fig:FC1}(f). Let $v$ be a vertex depicted in these figures. Since $\{s, t, v\}$ is a vertex cut of $C(m,n; k,l; c,d)$, then $C(m,n; k,l; c,d)-\{s, t, v\}$ is disconnected and contains three (or two) components in which two components (or one component) consist of only one vertex. Hence, any path between $s$ and $t$ must pass through $s$ or $v$ two times. Therefore, $C(m,n; k,l; c,d)$ contains no Hamiltonian $(s, t)$-path. For case (3) of condition (F8), consider Fig. \ref{fig:FC1}(g). A simple check shows that there is no Hamiltonian $(s,t)$-path in $C(m,n; k,l; c,d)$ containing both of vertices $u$ and $v$. For condition (F9), consider Fig. \ref{fig:FC1}(h). Let $v\in V(C(m,n; k,l; c,d)$ such that $v_x=1$ and $c+1 \leqslant v_y\leqslant c+l$. Since $a = 1$, $v$ is a cut vertex of $C(m,n; k,l; c,d)$. Obviously, any path between $s$ and $t$ must pass through $v$ two times. Therefore, $C(m,n; k,l; c,d)$ contains no Hamiltonian $(s, t)$-path.
\end{proof}
%%%%%%%%%%%%%%%%%%%%%%55

In the following, we will show that $C(m,n; k,l; c,d)$ always contains a Hamiltonian $(s, t)$-path when $(C(m,n; k,l; c,d),$ $s, t)$ does not satisfy conditions (F1), (F3), (F7), (F8), and (F9). We consider the case of $a = 1$ in Lemma \ref{HP-Cshaped1} and the case of $a\geqslant 2$ in Lemmas \ref{HP-Cshaped2} and \ref{HP-Cshaped3}.

\begin{figure}[!t]
\centering
\includegraphics[scale=0.9]{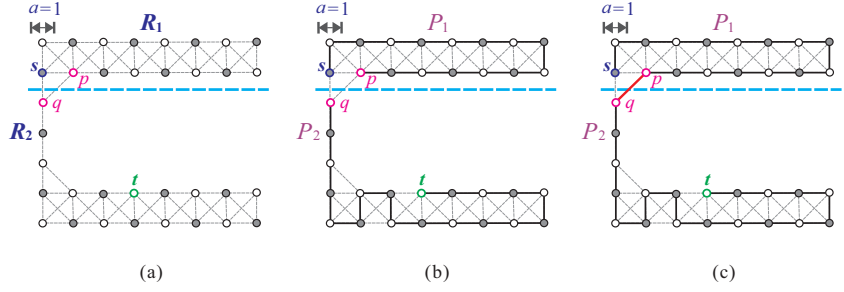}
\caption[]{(a) A horizontal separation on $C(m,n; k,l; c,d)$, (b) Hamiltonian paths $P_1$ and $P_2$ in $(R_1, s, p)$ and $(R_2, q, t)$, respectively, and (c) a Hamiltonian $(s, t)$-path in $C(m,n; k,l; c,d)$, where $a = 1$ and bold lines indicate the constructed Hamiltonian path.}
\label{fig:HP-SmallSize1}
\end{figure}

\begin{lem}\label{HP-Cshaped1}
Let $C(m, n; k, l; c, d)$ be a $C$-shaped supergrid graph with $a (= m - k )= 1$, and let $s$ and $t$ be its two distinct vertices such that $(C(m,n; k,l; c,d), s, t)$ does not satisfy conditions $\mathrm{(F1)}$, $\mathrm{(F3)}$, and $\mathrm{(F9)}$. Then, $C(m,n; k,l; c,d)$ contains a Hamiltonian $(s, t)$-path, i.e., $HP(C(m,n; k,l; c,d), s, t)$ does exist.
\end{lem}
\begin{proof}
Notice that, here, $(s_y\leqslant c$ and $t_y > c+l)$ or $(t_y\leqslant c$ and $s_y > c+l)$. If $s_y, t_y\leqslant c$, $s_y, t_y > c+l$, or $c+1\leqslant s_y (\mathrm{or}\ t_y)\leqslant c+l$, then $(C(m,n; k,l; c,d), s, t)$ satisfies condition (F1) or (F9). Without loss of generality, assume that $s_y\leqslant c$ and $t_y > c+l$. We make a horizontal separation on $C(m,n; k,l; c,d)$ to obtain two disjoint supergrid subgraphs $R_1 = R(m, c)$, $R_2 = L(m,n-c; k,l)$. Let $p\in V(R_1)$ and $q\in V(R_2)$ such that $p\thicksim q$, $q = (1, c+1)$, and $p = (1, c)$ if $s\neq (1, c)$; otherwise $p = (2, c)$ (see Fig. \ref{fig:HP-SmallSize1}(a)). Consider $(R_1, s, p)$. Condition (F1) holds, if

\begin{itemize}
\item [(i)] $c=1$, $k>1$, and $s_x\neq m$. Clearly, if this case holds, then $(C(m, n; k, l; c, d), s, t)$ satisfies condition (F1), a contradiction.
\item [(ii)] $c=2$ and $s_x=p_x\geqslant 2$. Clearly, in this case, $s_x=p_x=2$. It contradicts that $p=(1, c)$ when $s\neq (1, c)$.
\end{itemize}

\noindent Therefore, $(R_1, s, p)$ does not satisfy condition (F1). Now, consider $(R_2, q, t)$. Since $q_y = c+1$ and $t_y > c+l$, it is enough to show that $(R_2, q, t)$ does not satisfy condition (F1). Condition (F1) holds, if $d = 1$, $k > 1$, and $t_x\neq m$. Clearly, if this case holds, then $(C(m,n; k,l; c,d), s, t)$ satisfies condition (F1), a contradiction. Therefore, $(R_2, q, t)$ does not satisfy conditions (F1), (F3), and (F4). Since $(R_1, s, p)$ and $(R_2, q, t)$ do not satisfy conditions (F1), (F3), and (F4), by Lemma \ref{HamiltonianConnected-Rectangular} and Theorem \ref{HamiltonianConnected-L-shaped}, there exist Hamiltonian $(s, p)$-path $P_1$ and Hamiltonian $(q, t)$-path $P_2$ of $R_1$ and $R_2$, respectively (see Fig. \ref{fig:HP-SmallSize1}(b)). Then, $P = P_1 \Rightarrow P_2$ forms a Hamiltonian $(s, t)$-path of $C(m,n; k,l; c,d)$, as depicted in Fig. \ref{fig:HP-SmallSize1}(c).
\end{proof}
%%%%%%%%%%%%%%%%%%%%%%%%%%%%%%%%%%

\begin{lem}\label{HP-Cshaped2}
Let $C(m,n; k,l; c,d)$ be a $C$-shaped supergrid graph with $a (= m-k) > 1$, and let $s$ and $t$ be its two distinct vertices such that $(C(m,n; k,l; c,d), s, t)$ does not satisfy conditions $\mathrm{(F1)}$, $\mathrm{(F3)}$, $\mathrm{(F7)}$, and $\mathrm{(F8)}$. Assume that $c = d = 1$. Then, $C(m,n; k,l; c,d)$ contains a Hamiltonian $(s, t)$-path, i.e., $HP(C(m,n; k,l; c,d), s, t)$ does exist when $a\geqslant 2$ and $c = d = 1$.
\end{lem}
\begin{proof}
We prove this lemma by showing how to construct a Hamiltonian $(s, t)$-path of $C(m,n; k,l; c,d)$. Depending on whether $n = 3$, we consider the following cases:

Case 1: $n = 3$. Notice that, here, if $k>1$, then $s_x=t_x=m$. If $k>1$ and $s_x\neq m$ and $t_x\neq m$, then $(C(m,n; k,l; c,d), s, t)$ satisfies (F1) or (F3). Consider the positions of $s$ and $t$, there are the following two subcases:

\hspace{0.5cm}Case 1.1: $s_x > a$ or $t_x > a$. Without loss of generality, assume that $t_x > a$ and $t_y=1$. We make a vertical and horizontal separations on $C(m,n; k,l; c,d)$ to obtain two disjoint supergrid subgraphs $R_2 = R(k, c)$ and $R_1 = L(m,n; k,c+l)$, as depicted in Fig. \ref{fig:HP1}(a)--(c). Let $p\in V(R_1)$ and $q\in V(R_2)$ such that $p\thicksim q$, $q = (a+1, 1)$, and $p = (a, 1)$ if $s\neq (a, 1)$; otherwise $p = (a, 2)$ (see Fig. \ref{fig:HP1}(a)--(c)). Here, if $|R_2| = 1$ (i.e, $k = 1)$, then $q = t$. Consider $(R_1, s, p)$. Condition (F1) holds, if $a = 2$ and $s_y = p_y\geqslant 2$, or $p=(a, 2)$ and $s=(a, 3)$. Clearly, in any case, it contradicts that $p=(a, 1)$ when $s\neq (a, 1)$. Condition (F3) holds, if $k > 1$ and $s_x\leqslant a$. If this case holds, then $(C(m, n; k, l; c, d), s, t)$ satisfies (F3), a contradiction. Condition (F4) holds, if $n=3$, $a=2$, $k=1$, $s = (1, n)$, and $p = (a, 2)$. It contradicts that $p = (a, 1)$ when $s\neq (a, 1)$. Thus, $(R_1, s, p)$ does not satisfy conditions (F1), (F3), and (F4). Now, consider $(R_2,q,t)$. Since $q =(a+1, 1)$ and $t = (m, 1)$, it is clear that $(R_2, q, t)$ does not satisfy condition (F1). A Hamiltonian $(s, t)$-path of $C(m,n; k,l; c,d)$ can be constructed by similar arguments in proving Lemma \ref{HP-Cshaped1}, as shown in Fig. \ref{fig:HP1}(d).

\begin{figure}[h]
\centering
\includegraphics[scale=1]{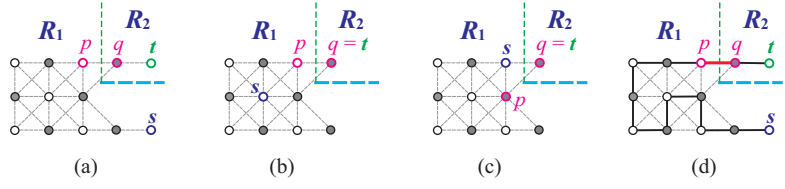}
\caption[]{(a)--(c) A vertical and horizontal separations on $C(m,n; k,l; c,d)$ under that $a\geqslant 2$, $c=d=1$, $n=3$, and $t_x > a$, and (d) a Hamiltonian $(s, t)$-path in $C(m,n; k,l; c,d)$ for (a), where bold lines indicate the constructed Hamiltonian path.}
\label{fig:HP1}
\end{figure}

\hspace{0.5cm}Case 1.2: $s_x, t_x\leqslant a$. In this subcase, $k = 1$ and hence $a = m-1$. If $k > 1$ and $s_x, t_x\leqslant a$, then $(C(m,n; k,l; c,d), s, t)$ satisfies (F3). Since $(C(m,n; k,l; c,d), s, t)$ does not satisfy condition (F8), $s_x\neq m-1$. Note that we assume that $s_x\leqslant t_x$. Then, there are two subcases:

\hspace{1.0cm}Case 1.2.1: $s_x, t_x < m-1$. Let $z_1 = (m-1, 1)$, $z_2 = (m-1, 2)$, and $z_3 = (m-1, 3)$. Let $e_1 = (z_1, z_2)$ and $e_2 = (z_2, z_3)$. We make a vertical and horizontal separations on $C(m,n; k,l; c,d)$ to obtain three disjoint rectangular supergrid subgraphs $R_1 = R(a, n)$, $R_2 = R(k, c)$, and $R_3 = R(k, d)$; as depicted in Fig. \ref{fig:HP2}(a). Assume that $v\in V(R_2)$ and $u\in V(R_3)$. By Lemma \ref{HP-3rectangle-boundary_path}, where $a>2$, or Lemma \ref{HamiltonianConnected-Rectangular}, where $a=2$, $R_1$ contains a Hamiltonian
$(s,t)$-path $P_1$ such that $e_1,e_2\in P_1$ (see Fig. \ref{fig:HP2}(b)). Then, $P_1$, $u$, $v$ can be combined into a Hamiltonian $(s, t)$-path of $C(m, n; k, l; c, d)$ by Statement (3) of Proposition \ref{Pro_Obs} (see Fig. \ref{fig:HP2}(c)).

\begin{figure}[h]
\centering
\includegraphics[scale=1]{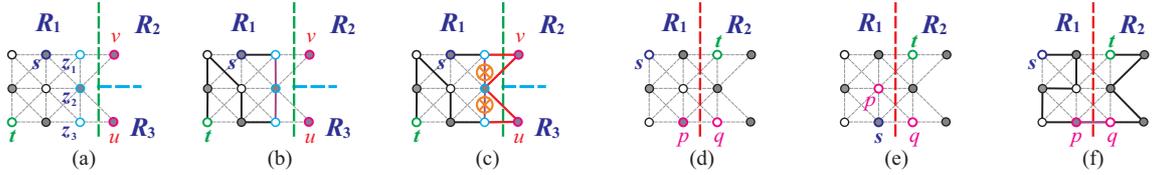}
\caption[]{(a) A vertical and horizontal separations on $C(m,n; k,l; c,d)$ under that $a\geqslant 2$, $c=d=1$, $n=3$, and $s_x, t_x < m-1$, (b) a Hamiltonian $(s, t)$-path in $R_1$ for (a), (c) a Hamiltonian $(s, t)$-path of $C(m,n; k,l; c,d)$ for (a), (d)--(e) a vertical separation on $C(m,n; k,l; c,d)$ for $a\geqslant 2$, $c=d=1$, $n=3$, $s_x< m-1$, and $t_x = m-1$, and (f) a Hamiltonian $(s, t)$-path of $C(m,n; k,l; c,d)$ for (d), where bold lines indicate the constructed Hamiltonian path and $\otimes$ represents the destruction of an edge while constructing a Hamiltonian cycle.}
\label{fig:HP2}
\end{figure}

\hspace{1.0cm}Case 1.2.2: $s_x < m-1$ and $t_x = m-1$. Since $(C(m,n; k,l; c,d), s, t)$ does not satisfy conditions (F1), (F7), and (F8), we have that (1) $t=(m-1, 1)$ or $t=(m-1,3)$ if $m-1>2$, and (2) $s_y=t_y=1$ or $s_y= t_y=3$ if $m-1=2$. Thus, $t_y = 1$ or $t_y = 3$. By symmetry, assume that $t = (m-1, 1)$. We make a vertical separation on $C(m,n; k,l; c,d)$ to obtain two disjoint supergrid subgraphs $R_1 = R( m-2, n)$ and $R_2 = C(2,n; k,l; c,d)$, as depicted in Fig. \ref{fig:HP2}(d) and \ref{fig:HP2}(e). Let $q\in V(R_2)$ and $p\in V(R_1)$ such that $p\thicksim q$, $q = (m-1, 3)$, and $p = (m-2, 3)$ if $s\neq (m-2, 3)$; otherwise $p = (m-2, 2)$ (see Fig. \ref{fig:HP2}(d)--(e)).

Notice that, here, if $s = (m-2, 3)$, then $m\geqslant 4$. If $m = 3$ and $s = (m-2, 3)$, then $(C(m,n; k,l; c,d), s, t)$ satisfies condition (F8). Consider $(R_2,q,t)$. Since $a=d=c=1$, $t=(m-1, 1)$, and $q=(m-1, 3)$, clearly $(R_2, q, t)$ does not satisfy conditions (F1), (F3), and (F9). Here, $(R_2, q, t)$ lies on Lemma \ref{HP-Cshaped1}. Now, consider $(R_1, s, p)$. Condition (F1) holds, if

\begin{itemize}
\item [(i)] $m-2 = 1$ and $ s = (m-2, 2)$. In this case, $(C(m,n; k,l; c,d), s, t)$ satisfies condition (F7), a contradiction.
\item [(ii)] $m-2=2$ and $s_y=p_y=2$. Clearly if $s_y=2$, then $p_y=3$. Hence, $s_y\neq p_y$.
\end{itemize}

\noindent Therefore, $(R_1, s, p)$ does not satisfy condition (F1). By Lemmas \ref{HamiltonianConnected-Rectangular} and \ref{HP-Cshaped1}, $R_1$ and $R_2$ contain Hamiltonian $(s, p)$-path $P_1$ and $(q, p)$-path $P_2$, respectively. Then, $P = P_1 \Rightarrow P_2$ forms a Hamiltonian $(s, t)$-path of $C(m,n; k,l; c,d)$, as depicted in Fig. \ref{fig:HP2}(f).

Case 2: $n > 3$. Since $n > 3$ and $c = d =1$, it follows that $l > 1$. We make a horizontal separation on $C(m,n; k,l; c,d)$ to obtain two disjoint
supergrid subgraphs $R_1 = L( m,n-n_2; k,l_1)$ and $R_2 = L(n_2,m; l_2,k)$, where $n_2 = c+l-1$, $l_1 = c+l-n_2$, and $l_2 =l-l_1$ (see Fig. \ref{fig:HP3}(a)--c)). Since $n_2 = c+l-1$ and $d = 1$, clearly $n - n_2 = 2$. Also since $n-n_2 = 2$ and $n > 3$, it follows that $n_2 \geqslant 2$. Depending on the positions of $s$ and $t$, there are the following two subcases:

\begin{figure}[h]
\centering
\includegraphics[scale=1]{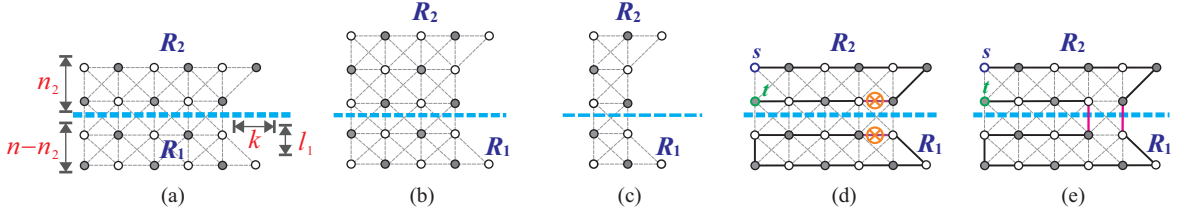}
\caption[]{(a)--(c) A horizontal separation on $C(m, n; k, l; c, d)$ for $a\geqslant 2$, $c=d=1$, and $n > 3$, (d) a Hamiltonian $(s,t)$-path in $R_2$ and a Hamiltonian cycle in $R_1$, and (e) a Hamiltonian $(s,t)$-path in $C(m,n; k,l; c,d)$.}
\label{fig:HP3}
\end{figure}

\hspace{0.5cm}Case 2.1: $s, t\in R_1$ or $s, t\in R_2$. Without loss of generality, assume that $s, t\in R_2$. Here, $k = 1$. If $k > 1$, then $(C(m,n; k,l; c,d), s, t)$ satisfies condition (F3).

\hspace{1.0cm}Case 2.1.1: $(n_2 > 2)$ or $(n_2 = 2$ and $[(s_x\neq t_x)$ or $(s_x=t_x=1)])$. Since $k=1$,  it is enough to show that $(R_2, s, t)$ does not satisfy conditions (F1) and (F4). Condition (F1) holds, if

\begin{itemize}
\item [(i)] $2\leqslant s_x=t_x\leqslant a$. By assumption this case does not occur.
\item [(ii)] $a = 2$ and $2\leqslant s_y=t_y\leqslant n_2$. Clearly if this case holds, then $(C(m,n; k,l; c,d), s, t)$ satisfies condition (F1), a contradiction.
\end{itemize}

\noindent Condition (F4) holds, if $a=2$, $n_2 > 2$, $s_x,t_x\leqslant 2$, $s_y,t_y\leqslant 2$, $s_x\neq t_x$, and $s_y\neq t_y$. Clearly if this case holds, then $(C(m,n; k,l; c,d), s, t)$ satisfies condition (F7), a contradiction. Therefore, $(R_2, s, t)$ does not satisfy conditions (F1), (F3), and (F4). Since $(R_2, s, t)$ does not satisfy conditions (F1), (F3), and (F4). By Theorem \ref{HamiltonianConnected-L-shaped}, $R_2$ contains a Hamiltonian $(s, t)$-path $P_2$ in which one edge $e_2$ is placed to face $R_1$. By Theorem \ref{HC-Lshaped}, $R_1$ contains a Hamiltonian cycle $HC_1$ such that its one flat face is placed to face $R_2$. Then, there exist two edges $e_1 \in HC_1$ and $e_2\in P_2$ such that $e_1\thickapprox e_2$ (see Fig. \ref{fig:HP3}(d)). By Statement (2) of Proposition \ref{Pro_Obs}, $P_2$ and $HC_1$ can be combined into a Hamiltonian $(s, t)$-path of $C(m,n; k,l; c,d)$. The construction of a such Hamiltonian path is depicted in Fig. \ref{fig:HP3}(e).

\hspace{1.0cm}Case 2.1.2: $n_2 = 2$ and $s_x,t_x > 1$. In this case, $n=4$, $a>2$, $l=2$, and $2\leqslant s_x=t_x\leqslant a-1$. If $s_x=t_x=a$, then $(C(m,n; k,l; c,d), s, t)$ satisfies condition (F1). Depending on whether $a=3$, we consider the following two subcases:

\begin{figure}[h]
\centering
\includegraphics[scale=1]{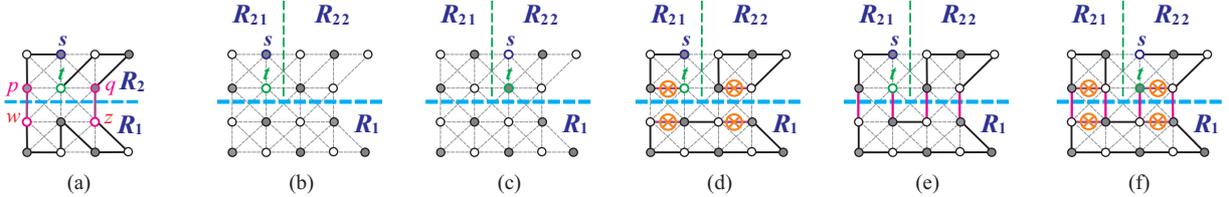}
\caption[]{(a) The pattern for constructing a Hamiltonian $(s, t)$-path in $C(m,n; k,l; c,d)$ under that $c=d=1$, $n=4$, $2\leqslant s_x=t_x\leqslant a-1$, and $a=3$, (b)--(c) a vertical and horizontal separations on $C(m, n; k, l; c, d)$ under that $c=d=1$, $n=4$, $2\leqslant s_x=t_x\leqslant a-1$, and $a>3$, (d) a Hamiltonian $(s,t)$-path in $R_{21}$ and Hamiltonian cycles in $R_1$ and $R_{22}$ for (b), (e) a Hamiltonian $(s,t)$-path in $C(m,n; k,l; c,d)$ for (b), and (f) a Hamiltonian $(s,t)$-path in $C(m,n; k,l; c,d)$ for (c).}
\label{fig:HP4}
\end{figure}

\hspace{1.5cm}Case 2.1.2.1: $a=3$. In this subcase, we can construct a Hamiltonian $(s,t)$-path by the pattern shown in Fig. \ref{fig:HP4}(a). Consider Fig. \ref{fig:HP4}(a). Since $w=(1,3)$ and $z=(3,3)$, clearly $(R_1, w, z)$ does not satisfy conditions (F1), (F3), and (F4). Thus by Theorem \ref{HamiltonianConnected-L-shaped}, $R_1$ contains a Hamiltonian $(w, z)$-path $P_1$. For $(R_2, s, t)$, we can construct two paths $P_{21}$ and $P_{22}$ such that $P_{21}$ connects $s$ and $p$, $P_{22}$ connects $q$ and $t$, $V(P_{21})\cap V(P_{22})=\emptyset$, $V(P_{21}\cup P_{22})=V(R_2)$, $p\thicksim w$, and $q\thicksim z$, as shown in Fig. \ref{fig:HP4}(a). Then, $P=P_{21} \Rightarrow P_1\Rightarrow P_{22}$ forms a Hamiltonian $(s, t)$-path of $C(m,n; k,l; c,d)$ (see Fig. \ref{fig:HP4}(a)).

\hspace{1.5cm}Case 2.1.2.2: $a>3$. We make a vertical separation on $R_2$ to obtain two disjoint subgraphs $R_{21}=R(m', n_2)$ and $R_{22}=L(n_2,m-m'; 1,1)$, where $m'=s_x$ if $s_x<a-1$; otherwise $m'=s_x-1$; as shown in Fig. \ref{fig:HP4}(b)--(c). First, let $s,t\in R_{21}$ and consider Fig. \ref{fig:HP4}(b). Clearly since $s_x=t_x=m'$, $(R_{21}, s, t)$ does not satisfy condition (F1). By Lemma \ref{HamiltonianConnected-Rectangular}, $R_{21}$ contains a canonical Hamiltonian $(s,t)$-path $P_{21}$. Then, there exists one edge $e_{21} \in P_{21}$ that is placed to face $R_1$. By Theorem \ref{HC-Lshaped}, $R_{22}$ and $R_1$ contain Hamiltonian cycle $HC_{22}$ and $HC_1$, respectively. Using the algorithm of \cite{Keshavarz19a}, we can construct $HC_{22}$ and $HC_1$ to satisfy that one flat face of $HC_1$ is placed to face $R_{21}$ and $R_{22}$, and one flat face of $HC_{22}$ is placed to face $R_1$ (see Fig. \ref{fig:HP4}(d)). Then, there exist four edges $e_1, e_2\in HC_1$, $e_{21}\in P_{21}$, and $e_{22}\in HC_{22}$ such that $e_1\thickapprox e_{21}$ and $e_2\thickapprox e_{22}$. By Statements (1) and (2) of Proposition \ref{Pro_Obs}, $P_{21}$, $HC_1$, and $HC_{22}$ can be combined into a Hamiltonian $(s, t)$-path of $C(m, n; k, l; c,d)$. The construction of a such Hamiltonian path is depicted in Fig. \ref{fig:HP4}(e). Now, let $s, t\in R_{22}$ and consider Fig. \ref{fig:HP4}(c). Since $s_x=t_x=1$ and $s_x,t_x < a$ (in $R_{22}$), it is clear that $(R_{22}, s, t)$ does not satisfy conditions (F1), (F3), and (F4). By similar arguments in proving $s,t\in R_{21}$, a Hamiltonian $(s, t)$-path of $C(m, n; k, l; c,d)$ can be constructed (see Fig. \ref{fig:HP4}(f)).
%By Theorem \ref{HamiltonianConnected-L-shaped}, $R_{22}$ contains a Hamiltonian $(s, t)$-path. Using the algorithm of \cite{Keshavarz19a}, we can %construct a Hamiltonian $(s, t)$-path $P_1$ in which one edge $e_1$ is placed to face $R_1$. By Lemmas \ref{HC-rectangular_supergrid_graphs} and %\ref{HC-Lshaped}, $R_{21}$ and $R_1$ contain Hamiltonian cycle $HC_{21}$ and $HC_1$, respectively. We can place one flat faces of $HC_1$ to %respectively face $R_{21}$ and $R_{22}$. Then, there exist four edges $e_1, e_2\in  CH_2$, $e_3\in P$, and $e_4\in CH_1$ such that $e_1\approx e_3$ %and $e_2 \approx e_4$; as shown in Fig. \ref{fig:HP4}(d). By Statement (2) of Proposition \ref{Pro_Obs},
%$P$, $HC_1$, and $CH_2$ can be combined into a Hamiltonian $(s, t)$-path of $C(m, n; k, l;c,d)$.
%The construction of a such Hamiltonian path is depicted in Fig. \ref{fig:HP4}(f).

\hspace{0.5cm}Case 2.2: ($s\in R_1$ and $t\in R_2$) or ($t\in R_1$ and $s\in R_2$). Without loss of generality, assume that $t\in R_1$ and $s\in R_2$. Let $q\in V(R_1)$ and $p\in V(R_2)$ such that $p\thicksim q$, where

$$\begin{cases}
p=(1, n_2)\ \mathrm{and}\ q=(1, n_2+1),   &  \mathrm{if}\ s\neq (1, n_2)\ \mathrm{and}\ t\neq (1, n_2+1);\\
p=(2, n_2)\ \mathrm{and}\ q=(2, n_2+1),   &  \mathrm{if}\ s= (1, n_2)\ \mathrm{and}\ t= (1, n_2+1);\\
p=(2, n_2)\ \mathrm{and}\ q=(1, n_2+1),   &  \mathrm{if}\ s= (1, n_2)\ \mathrm{and}\ t\neq (1, n_2+1);\\
p=(1, n_2)\ \mathrm{and}\ q=(2, n_2+1),   &  \mathrm{otherwise.}\\
\end{cases}$$

\noindent Consider $(R_2, s, p)$ and $(R_1, q, t)$. Condition (F1) holds, if

\begin{itemize}
\item [(i)] $(n_2=2$ and $s_x=p_x=2)$ or $(n-n_2=2$ and $q_x=t_x=2)$. Obviously in this case, $s_x=p_x=2$ and $q_x=t_x=2$. It contradicts that $p_x=q_x=1$ when $s_x\neq 1$ and $t_x\neq 1$.
\item [(ii)] $k > 1$ and $a < s_x,t_x < m$. If this case holds, then $(C(m,n; k,l; c,d), s, t)$ satisfies condition (F1), a contradiction.
\end{itemize}

\noindent Condition (F3) holds, if $k > 1$ and $s_x, t_x\leqslant a$. If this case holds, then $(C(m,n; k,l; c,d), s, t)$ satisfies condition (F3), a contradiction. Condition (F4) holds, if $a=2$ and $[(n_2 > 2$ and $s_y, p_y\leqslant 2)$ or $(n-n_2 > 2$ and $q_y, t_y\geqslant n-1)]$. A simple check shows that these cases do not occur. Therefore, $(R_2, s, p)$ and $(R_1, q, t)$ do not satisfy conditions (F1), (F3), and (F4). A Hamiltonian $(s, t)$-path of $C(m,n; k,l; c,d)$ can be constructed by similar arguments in proving Lemma \ref{HP-Cshaped1}. Notice that, here, $R_1$ and $R_2$ are $L$-shaped supergrid graphs.
\end{proof}
%%%%%%%%%%%%%%%%%%%%%%%%%%%%%%%%%%

\begin{lem}\label{HP-Cshaped3}
Let $C(m,n; k,l; c,d)$ be a $C$-shaped supergrid graph with $a (= m-k) > 1$, and let $s$ and $t$ be its two distinct vertices such that $(C(m,n; k,l; c,d), s, t)$ does not satisfy conditions $\mathrm{(F1)}$, $\mathrm{(F3)}$, $\mathrm{(F7)}$, and $\mathrm{(F8)}$. Assume that $c > 1$ or $d > 1$. Then, $C(m,n; k,l; c,d)$ contains a Hamiltonian $(s, t)$-path, i.e., $HP(C(m,n; k,l; c,d), s, t)$ does exist when $a\geqslant 2$ and $(c > 2$ or $d > 2)$.
\end{lem}
\begin{proof}
Without loss of generality, assume that $d > 1$. Since $d > 1$, $c\geqslant 1$, and $l\geqslant 1$, thus $n\geqslant 4$. We make a horizontal separation on $C(m,n; k,l; c,d)$ to obtain two disjoint supergrid subgraphs $R_2 = L(n_2,m; l,k)$ and $R_1 = R(m, n-n_2)$, where $n_2 = c+l$ and $n-n_2=d$ (see Fig. \ref{fig:HP5}(a)--(c)). Since $n\geqslant 4$, $n_2 = c+l$, $c, l\geqslant 1$, and $d > 1$, it follows that $n-n_2, n_2\geqslant 2$. Depending on the positions of $s$ and $t$, there are the following three subcases:

\begin{figure}[h]
\centering
\includegraphics[width=1.0\textwidth]{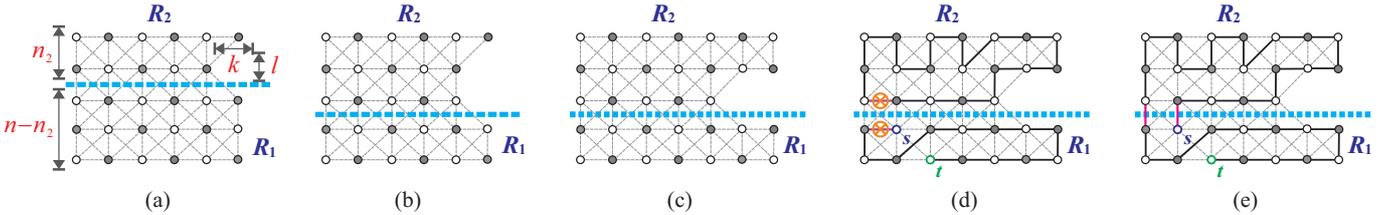}
\caption[]{(a)--(c) A horizontal separation on $C(m,n; k,l; c,d)$ under that $a > 1$ and $d > 1$, (d) a Hamiltonian $(s, t)$-path in $R_1$ and a Hamiltonian cycle in $R_2$ for $s, t\in R_1$ and $\{s, t\}$ is not a vertex cut of $R_1$, and (e) a Hamiltonian $(s,t)$-path in $C(m,n; k,l; c,d)$ for (d).}
\label{fig:HP5}
\end{figure}

Case 1: $s,t\in R_1$. In this case, ($c > 1$) or ($c = 1$ and $k = 1$). If $c = 1$ and $k > 1$, then $(C(m,n; k,l; c,d), s, t)$ satisfies condition (F3), a contradiction. Depending on the size of $n-n_2$, we consider the following two subcases:

\hspace{0.5cm}Case 1.1: ($n-n_2 > 2$) or ($n-n_2=2$ and $[(s_x\neq t_x)$, $(s_x=t_x=1)$, or $(s_x=t_x=m)]$). In this subcase, $\{s, t\}$ is not a vertex cut of $R_1$. Then, $(R_1, s, t)$ does not satisfy condition (F1). By Lemma \ref{HamiltonianConnected-Rectangular}, $R_1$ contains a canonical Hamiltonian $(s,t)$-path $P_1$. Using the algorithm of \cite{Hung17a}, we can construct a Hamiltonian $(s,t)$-path $P_1$ of $R_1$ in which one edge $e_1$ is placed to face $R_2$. By Theorem \ref{HC-Lshaped}, $R_2$ contains a Hamiltonian cycle $HC_2$. Using the algorithm of \cite{Keshavarz19a}, we can construct $HC_2$ such that its one flat face is placed to face $R_1$. Then, there exist two edges $e_1 \in P_1$ and $e_2\in HC_2$ such that $e_1 \thickapprox e_2$ (see Fig. \ref{fig:HP5}(d)). By Statement (2) of Proposition \ref{Pro_Obs}, $P_1$ and $HC_2$ can be combined into a Hamiltonian $(s, t)$-path of $C(m,n; k,l; c,d)$. The construction of a such Hamiltonian path is depicted in Fig. \ref{fig:HP5}(e).

\begin{figure}[h]
\centering
\includegraphics[scale=1.0]{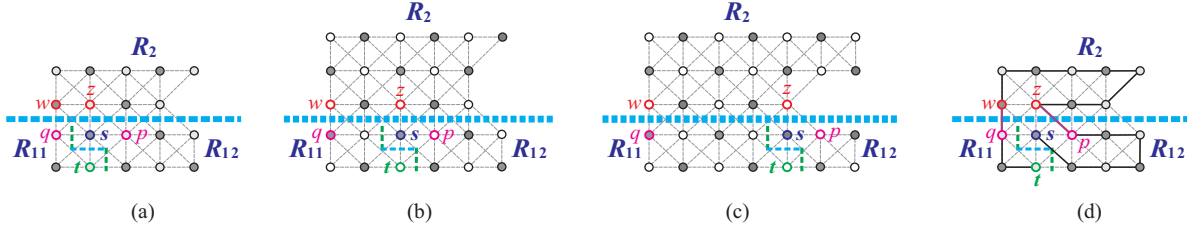}
\caption[]{(a)--(c) A vertical and horizontal separations on $R_1$ under that $s, t\in R_1$ and $\{s, t\}$ is a vertex cut of $R_1$, and (d) a Hamiltonian $(s,t)$-path in $C(m,n; k,l; c,d)$ for (a).}
\label{fig:HP6}
\end{figure}

\hspace{0.5cm}Case 1.2: $n-n_2 = 2$ and $1 < s_x=t_x < m$. In this subcase, $\{s, t\}$ is a vertex cut of $R_1$. Notice that, here, $s_x, t_x\leqslant a$. If $s_x, t_x > a$ and $k > 1$, then $(C(m,n; k,l; c,d), s, t)$ satisfies condition (F1). Without loss of generality, assume that $s_y < t_y$. We make a vertical and horizontal separations on $R_1$ to obtain two disjoint supergrid subgraphs $R_{11} = L(s_x, n-n_2; 1,1)$ and  $R_{12} = R_1\setminus R_{11}$. Let $p\in V(R_{12})$, $q\in V(R_{11})$, and $w, z\in V(R_2)$ such that $p\thicksim z$, $q\thicksim w$, $q = (1, n_2+1)$, $w = (1, n_2)$, $p = (s_x+1, n_2+1)$, and $z = (s_x, n_2)$ (see Fig. \ref{fig:HP6}(a)--(c)). A simple check shows that $(R_{12}, s, p)$, $(R_2, z, w)$, and $(R_{11}, q, t)$ do not satisfy conditions (F1), (F3), and (F4). Thus, by Theorem \ref{HamiltonianConnected-L-shaped}, $R_{12}$, $R_2$, and $R_{11}$ contain a Hamiltonian $(s, p)$-path $P_{12}$, Hamiltonian $(z, w)$-path $P_2$, and Hamiltonian $(q, t)$-path $P_{11}$, respectively. Then, $P = P_{12}\Rightarrow P_2\Rightarrow P_{11}$ forms a Hamiltonian $(s, t)$-path of $C(m,n; k,l; c,d)$. The construction of a such Hamiltonian path is depicted in Fig. \ref{fig:HP6}(d).

Case 2: $s,t\in R_2$. A Hamiltonian $(s, t)$-path of $(C(m,n; k,l; c,d)$ can be constructed by similar arguments in proving Case 2.1 of Lemma \ref{HP-Cshaped2} (see Fig. \ref{fig:HP7}). Notice that, in this case, $R_1$ is a rectangular supergrid graph.

\begin{figure}[h]
\centering
\includegraphics[scale=1.0]{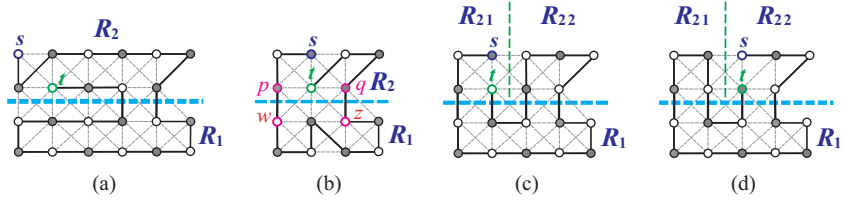}
\caption[]{A Hamiltonian $(s,t)$-path in $C(m,n;k,l;c,d)$ under that $a, d > 1$ and $s,t\in R_2$.}
\label{fig:HP7}
\end{figure}

Case 3: ($s\in R_1$ and $t\in R_2$) or ($s\in R_2$ and $t\in R_1$). A Hamiltonian $(s, t)$-path of $(C(m,n; k,l; c,d)$ can be constructed by similar arguments in proving Case 2.2 of Lemma \ref{HP-Cshaped2}. Notice that, in this case, $R_1$ is a rectangular supergrid graphs.
\end{proof}

From Lemmas \ref{Necessary-condition-Cshaped}--\ref{HP-Cshaped3}, it immediately follows that the following theorem holds true.

\begin{thm}\label{HP-Theorem-Cshaped}
Let $C(m,n; k,l; c,d)$ be a $C$-shaped supergrid graph with vertices $s$ and $t$. $C(m,n; k,l; c,d)$ contains a Hamiltonian $(s, t)$-path, i.e., $HP(C(m,n; k,l; c,d), s, t)$ does exist if and only if $(C(m,n; k,l; c,d), s, t)$ does not satisfy conditions $\mathrm{(F1)}$, $\mathrm{(F3)}$, $\mathrm{(F7)}$, $\mathrm{(F8)}$, and $\mathrm{(F9)}$.
\end{thm}

%====================================================
\section{The longest $(s, t)$-paths in $C$-shaped supergrid graphs}\label{Sec_Algorithm}
%====================================================
From Theorem \ref{HP-Theorem-Cshaped}, it follows that if $(C(m,n; k,l; c,d), s, t)$ satisfies one of the conditions (F1), (F3), (F7), (F8), and (F9), then $(C(m,n; k,l; c,d), s, t)$ contains no Hamiltonian $(s, t)$-path. So in this section, first for these cases we give upper bounds on the lengths of longest paths between $s$ and $t$. Then, we show that these upper bounds is equal to the length of longest paths between $s$ and $t$. Notice that the isomorphic cases are omitted. The following lemmas give these bounds.
%%%%%%%%%%%%%%%%%%%%%%%%%

We first consider the case of $a = 1$. In this case, $(C(m,n; k,l; c,d), s, t)$ may satisfy condition (F1), (F3), or (F9). We compute the upper bound of the longest $(s, t)$-path in this case as the following lemma.

\begin{figure}[h]
\centering
\includegraphics[scale=1]{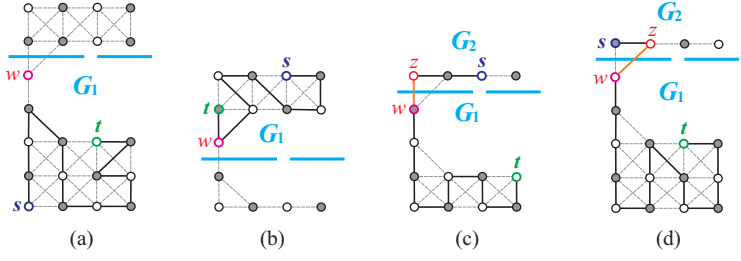}
\caption[]{A longest path between $s$ and $t$ under that $a = 1$ for (a)--(b) (FC7) holds, and (c)--(d) (FC8) holds, where bold lines indicate the constructed longest $(s, t)$-path and bold dash lines indicate the separations.}
\label{fig:FC7-FC8}
\end{figure}

\begin{lem}\label{Lemma:FC7-FC8}
Let $a = 1$ and $w = (1, c+1)$. Assume that $C(m,n; k,l; c,d),s,t)$ satisfies one of the conditions $(\mathrm{F1)}$, $(\mathrm{F3)}$, and $(\mathrm{F9)}$. Then, the following conditions hold:
\begin{description}
  \item[$\mathrm{(FC7)}$] If $s_y, t_y > c$ (resp., $s_y, t_y \leqslant c$), then the length of any path between $s$ and $t$ cannot exceed $\hat{L}(G_1, s, t)$, where $G_1 = L(m,n-c; k,l)$ (resp., $G_1= L(c+1,m; 1,k)$) (see Fig. \emph{\ref{fig:FC7-FC8}(a)--(b)}).
  \item[$\mathrm{(FC8)}$] If $(s_y\leqslant c$ and $t_y > c)$ or $(t_y\leqslant c$ and $s_y > c)$, without loss of generality assume that $s_y\leqslant c$, then the length of any path between $s$ and $t$ cannot exceed $\hat{L}(G_2, s, z)+\hat{L}(G_1, w, t)$, where $G_2 = R(m, c)$, $G_1 = L(m,n-c; k,l)$, and $z=(1, c)$ if $s\neq (1, c)$; otherwise $z=(2, c)$ (see Fig. \emph{\ref{fig:FC7-FC8}(c)--(d)}).
\end{description}
\end{lem}
\begin{proof}
The proof is straightforward, see Fig. \ref{fig:FC7-FC8}.
\end{proof}

Next, we consider the case of $a\geqslant 2$. In this case, $(C(m,n; k,l; c,d), s, t)$ may satisfy condition (F1), (F3), (F7), or (F8). Depending on the sizes of $c$ and $d$, we consider the subcases of (1) $c\geqslant 2$ and $d\geqslant 2$, and (2) $c = 1$ or $d = 1$. Consider that $c\geqslant 2$ and $d\geqslant 2$. Then, $(C(m,n; k,l; c,d), s, t)$ does not satisfy conditions (F3), (F7), and (F8). Thus, $(C(m,n; k,l; c,d), s, t)$ may satisfy condition (F1) only. We can see that $s$ or $t$ is not a cut vertex when $a, c, d\geqslant 2$. Thus, $\{s, t\}$ is a vertex cut when $(C(m,n; k,l; c,d), s, t)$ satisfies condition (F1). We can see from the structure of $C(m,n; k,l; c,d)$ that $a$, $c$, or $d$ is equals to 2 if $\{s, t\}$ is a vertex cut. The following lemma shows the upper bound of the longest $(s, t)$-path under that $a, c, d\geqslant 2$ and $(C(m,n; k,l; c,d), s, t)$ satisfies condition (F1).

\begin{figure}[h]
\centering
\includegraphics[scale=1]{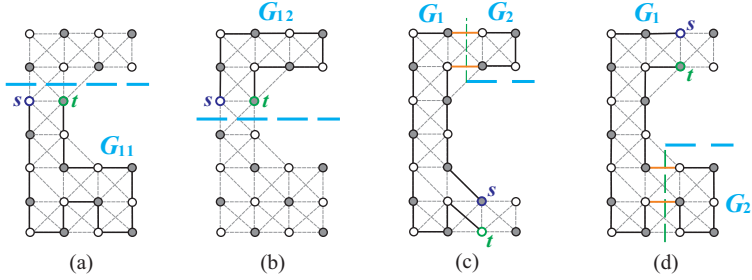}
\caption[]{A longest path between $s$ and $t$ under that $a\geqslant 2$ and $c, d\geqslant 2$ for (a)--(b) (FC9) holds, and (c)--(d) (FC10) holds, where bold lines indicate the constructed longest $(s, t)$-path.}
\label{fig:FC9-FC10}
\end{figure}

\begin{lem} \label{Lemma:FC9-FC10} Assume that $a,c, d\geqslant 2$ and $\{s, t\}$ is a vertex cut. Then, the following conditions hold:
\begin{description}
  \item[$\mathrm{(FC9)}$] If $a = 2$ and $c+1\leqslant s_y=t_y\leqslant c+l$, then the length of any path between $s$ and $t$ cannot exceed $\max\{\hat{L}(G_{11}, s, t), \hat{L}(G_{12}, s, t)\}$, where $G_{11} = L(m,n'; k,l')$, $G_{12} = L(s_y,m; k',k)$, $n'=n-s_y+1$, $l'=n'-d$, and $k'=s_y-1$ (see Fig. \emph{\ref{fig:FC9-FC10}(a)--(b)}).
  \item[$\mathrm{(FC10)}$] If $k\geqslant 2$, $a+1\leqslant s_x, t_x\leqslant a+k-1(=m-1)$, and $[(c = 2$ and $s_y, t_y \leqslant c)$ or $(d = 2$, and $s_y, t_y \geqslant c+l)]$, without loss of generality assume that $d = 2$ and $s_y, t_y \geqslant c+l$, then the length of any path between $s$ and $t$ cannot exceed $\hat{L}(G_1, s, t)+|G_2|=\hat{L}(G_1, s, t)+ k\times c$, where $G_1 = L(m,n; k,l+c)$ and $G_2 = R(k, c)$ (see Fig. \emph{\ref{fig:FC9-FC10}(c)--(d)}).
\end{description}
\end{lem}
\begin{proof}
Consider Fig. \ref{fig:FC9-FC10}(a)--(b). Removing $s$ and $t$ clearly disconnects $C(m,n; k,l; c,d)$ into two components $G_1$ and $G_2$. Thus, a simple path between $s$ and $t$ can only go through one of these components. Therefore, its length cannot exceed the size of the largest component. Notice that, for (FC10), the length of any path between $s$ and $t$ is equal to $\max\{\hat{L}(G_1,s,t)+|G_2|,2\times (m-s_x+1)\}$ (see Fig. \ref{fig:FC9-FC10}(c)--(d)). Since $a, c>1$, it is obvious that the length of any path between $s$ and $t$ cannot exceed $\hat{L}(G_1, s, t)+|G_2| = \hat{L}(G_1,s,t)+k\times c$.
\end{proof}

%\newpage

\begin{figure}[h]
\centering
\includegraphics[scale=1]{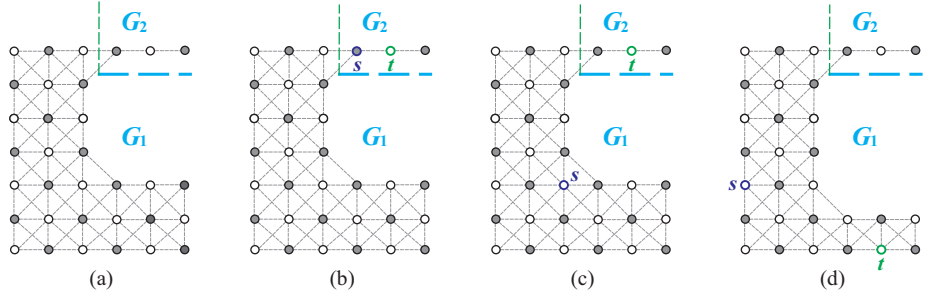}
\caption[]{(a) The separations on $C(m,n; k,l; c,d)$ for $a\geqslant 2$ and $c = 1$, (b) the case of $s, t\in G_2$, (c) the case of $s\in G_1$ and $t\in G_2$, and (d) the case of $s, t\in G_1$.}
\label{fig:c1-cases}
\end{figure}

In the following, we will consider that $a\geqslant 2$, and ($c=1$ or $d=1$). Without loss of generality, assume that $c = 1$. We first make a horizontal and vertical separations on $C(m,n; k,l; c,d)$ to obtain two disjoint subgraphs $G_1=L(m, n; k, l+c)$ and $G_2=R(m-a, c)$, as depicted in Fig. \ref{fig:c1-cases}(a), where $a\geqslant 2$, $c = 1$, and $G_2$ is a path graph. Depending the locations of $s$ and $t$, we consider the following cases:

Case I: $s, t\in G_2$. In this case, $k > 1$, $s_y=t_y=1$, and $a+1\leqslant s_x, t_x\leqslant m$. Then, $(C(m,n; k,l; c,d), s, t)$ may satisfy condition (F1) or (F3), as depicted in Fig. \ref{fig:c1-cases}(b).

Case II: $s\in G_1$ and $t\in G_2$. In this case, $(C(m,n; k,l; c,d), s, t)$ may satisfy condition (F1) or (F3), as depicted in Fig. \ref{fig:c1-cases}(c).

Case III: $s, t\in G_1$. In this case, $(C(m,n; k,l; c,d), s, t)$ may satisfy condition (F1), (F3), (F7), or (F8), as depicted in Fig. \ref{fig:c1-cases}(d). If $d=1$, $s_y=t_y=n$, and $a+1\leqslant s_x, t_x\leqslant m$, then it is the same as Case I. Depending on whether $(G_1, s, t)$ satisfies condition (F1), there are the following three subcases:

\hspace{0.5cm}Case III.1: $s$ or $t$ is a cut vertex of $G_1$. In this subcase, $d = 1$ and it is the same as Case I or Case II.

\begin{figure}[h]
\centering
\includegraphics[scale=1]{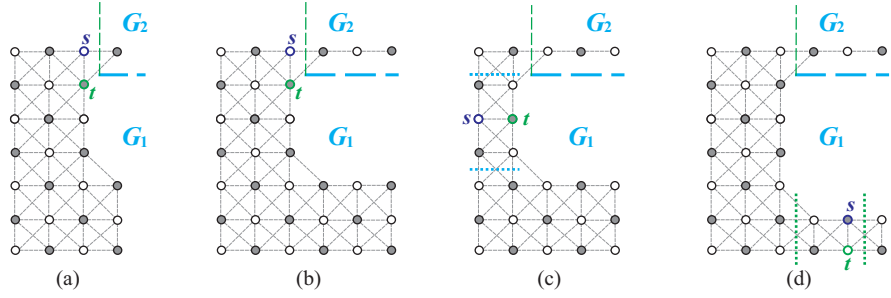}
\caption[]{The cases for $s, t\in G_1$ and $\{s, t\}$ is a vertex cut of $G_1$, where (a)--(b) $\{s, t\} = \{(a, 1), (a, 2)\}$, (c) $a = 2$, $s_y = t_y$ and $a+c\leqslant s_x,t_x\leqslant a+l$, and (d) $d = 2$, $k > 1$, $a < s_x = t_x < m $ and $n-1\leqslant s_y, t_y\leqslant n$.}
\label{fig:c1-cases-vertex-cut}
\end{figure}

\hspace{0.5cm}Case III.2: $\{s, t\}$ is a vertex cut of $G_1$. Consider the following subcases:

\hspace{1.0cm}Case III.2.1: $\{s, t\} = \{(a, 1), (a, 2)\}$. In this subcase, $k=1$ or $k>1$, as shown in Fig. \ref{fig:c1-cases-vertex-cut}(a)--(b). If $k > 1$, then $(C(m,n; k,l; c,d), s, t)$ satisfies conditions (F1) and (F3); otherwise it satisfies condition (F1).

\hspace{1.0cm}Case III.2.2: $a = 2$, $s_y = t_y$, and $c+1\leqslant s_x,t_x\leqslant c+l$. In this subcase, it is similar to condition (FC9) in Lemma \ref{Lemma:FC9-FC10} (see Fig. \ref{fig:c1-cases-vertex-cut}(c)).

\hspace{1.0cm}Case III.2.3: $d = 2$, $k > 1$, $a < s_x = t_x < m$, and $n-1\leqslant s_y, t_y\leqslant n$. In this subcase, it is similar to condition (FC10) in Lemma \ref{Lemma:FC9-FC10} (see Fig. \ref{fig:c1-cases-vertex-cut}(d)).

\begin{figure}[h]
\centering
\includegraphics[width=0.98\textwidth]{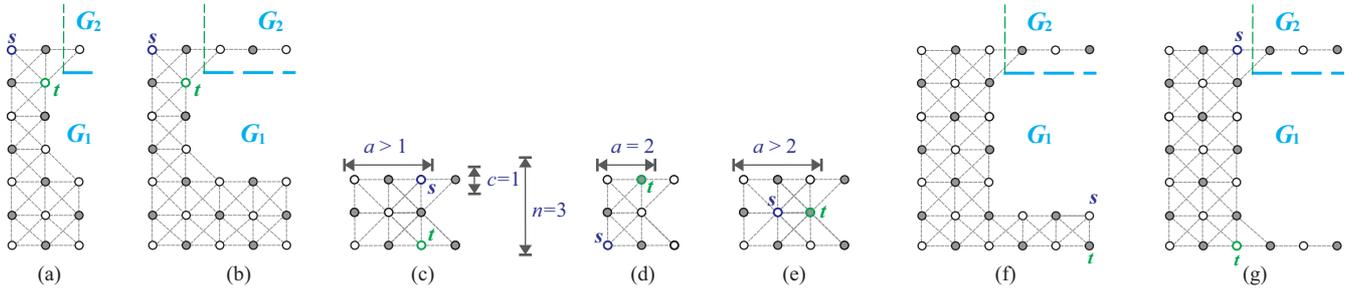}
\caption[]{The cases for $s, t\in G_1$ and $\{s, t\}$ is not a vertex cut of $G_1$, where (a)--(b) $a=2$, $s_y, t_y\leqslant 2$, $s_y\neq t_y$, and $s_x\neq t_x$, (c)--(e) $(C(m,n; k,l; c,d), s, t)$ satisfies condition (F8), and (f)--(g) $(G_1, s, t)$ satisfies condition (F3) but it does not satisfy condition (F1).}
\label{fig:c1-cases-non-vertex-cut}
\end{figure}

\hspace{0.5cm}Case III.3: $(G_1, s, t)$ does not satisfy condition (F1). In this subcase, $s$ and $t$ are not cut vertices and $\{s, t\}$ is not a vertex cut of $G_1$. Then, $(C(m,n; k,l; c,d), s, t)$ may satisfy condition (F3), (F7), or (F8). Depending on the size of $k$, we consider the following subcases:

\hspace{1.0cm}Case III.3.1: $k = 1$. In this subcase, $(C(m,n; k,l; c,d), s, t)$ may satisfy condition (F7) and (F8) as follows:

\hspace{1.5cm}Case III.3.1.1: $(C(m,n; k,l; c,d), s, t)$ satisfies condition (F7). In this subcase, $a=2$ and $[(c = 1$ and $\{s, t\} = \{(1, 1), (2, 2)\}$ or $\{(1, 2), (2, 1)\})$ or $(d = 1$ and $\{s, t\} = \{(1, n), (2, n-1)\}$ or $\{(1, n - 1), (2, n)\})]$ (see Fig. \ref{fig:c1-cases-non-vertex-cut}(a)).

\hspace{1.5cm}Case III.3.1.2: $(C(m,n; k,l; c,d), s, t)$ satisfies condition (F8). In this subcase, $n=3$ and $k=c=d=1$ (see Fig. \ref{fig:c1-cases-non-vertex-cut}(c)--(e)).

\hspace{1.0cm}Case III.3.2: $k > 1$. In this subcase, $(C(m,n; k,l; c,d), s, t)$ satisfies condition (F3) but it does not satisfy condition (F1). There are the following subcases:

\hspace{1.5cm}Case III.3.2.1: $a=2$ and $[(c = 1$ and $\{s, t\} = \{(1, 1), (2, 2)\}$ or $\{(1, 2), (2, 1)\})$ or $(d = 1$ and $\{s, t\} = \{(1, n), (2, n-1)\}$ or $\{(1, n - 1), (2, n)\})]$. In this subcase, it is the same as Case III.3.1.1 (see Fig. \ref{fig:c1-cases-non-vertex-cut}(a)).

\hspace{1.5cm}Case III.3.2.2: $(c = 1$ and $\{s, t\} \neq \{(1, 1), (2, 2)\}$ and $\{(1, 2), (2, 1)\})$ or $(d = 1$ and $\{s, t\} \neq \{(1, n), (2, n-1)\}$ and $\{(1, n - 1), (2, n)\})$ (see Fig. \ref{fig:c1-cases-non-vertex-cut}(e)--(f)).\\

Based on the above cases, we compute the upper bounds of longest $(s, t)$-paths on $(C(m,n; k,l; c,d), s, t)$ under that $a\geqslant 2$ and $c = 1$ as the following lemma.

\begin{figure}[h]
\centering
\includegraphics[scale=1]{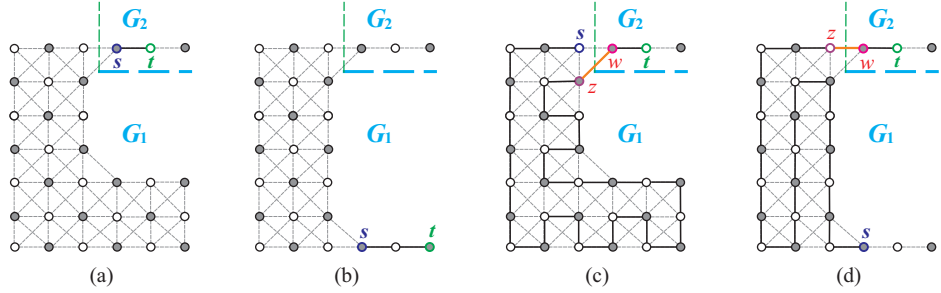}
\caption[]{The longest path between $s$ and $t$ under that $a\geqslant 2$ and $c = 1$, where (a)--(b) (FC11) holds and (c)--(d) (FC12) holds, where bold lines indicate the constructed longest $(s, t)$-path.}
\label{fig:FC11-FC12}
\end{figure}

\begin{figure}[h]
\centering
\includegraphics[width=0.98\textwidth]{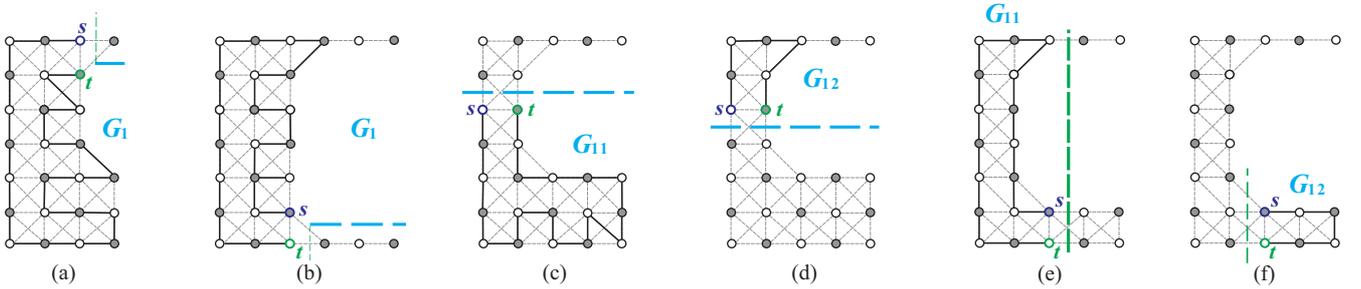}
\caption[]{The longest path between $s$ and $t$ under that $a\geqslant 2$ and $c = 1$, where (a)--(b) (FC13) holds, (c)--(d) (FC14) holds, and (e)--(f) (FC15) holds.}
\label{fig:FC13-FC15}
\end{figure}

\begin{figure}[h]
\centering
\includegraphics[width=0.98\textwidth]{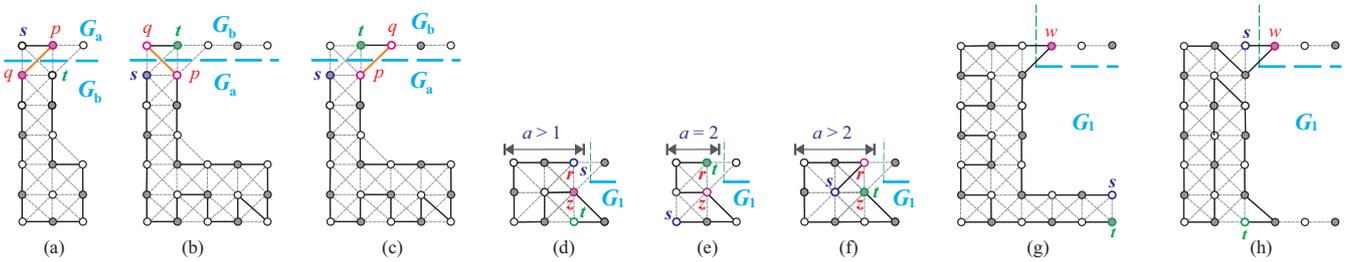}
\caption[]{The longest path between $s$ and $t$ under that $a\geqslant 2$ and $c = 1$, where (a)--(c) (FC16) holds, (d)--(f) (FC17) holds, and (g)--(h) (FC18) holds.}
\label{fig:FC16-FC18}
\end{figure}

\begin{lem} \label{Lemma:FC11-FC18} Assume that $a\geqslant 2$ and $c = 1$. Let $w = (a+1, 1)$. Then, the following conditions hold:
\begin{description}
  \item[$\mathrm{(FC11)}$] If $k > 1$, $a+1\leqslant s_x, t_x\leqslant m$, and ($s_y=t_y=1$ or $s_y=t_y=n$), then the length of any path between $s$ and $t$ cannot exceed $t_x-s_x+1$ (see Fig. \emph{\ref{fig:FC11-FC12}(a)--(b)}, and refer to Case I and Case III.1).
  \item[$\mathrm{(FC12)}$] If $t_x > 1$, $t_y=1$, and $[(s_x\leqslant a)$ or $(s_x > a$ and $s_y > c+l)]$, then the length of any path between $s$ and $t$ cannot exceed $\hat{L}(G_1, s, z)+ \hat{L}(G_2, w, t)$, where $G_1 = L(m,n; k,l+c)$, $G_2 = R(k, c)$, and $z = (a, 1)$ if $s\neq (a, 1)$; otherwise $z = (a, 2)$ (see Fig. \emph{\ref{fig:FC11-FC12}(c)--(d)}, and refer to Case II and Case III.1).
  \item[$\mathrm{(FC13)}$] If $\{s, t\} = \{(a, 1), (a, 2)\}$ (resp., $d=1$ and $\{s, t\} = \{(a, n-1), (a, n)\}$), then the length of any path between $s$ and $t$ cannot exceed $\hat{L}(G_1, s, t)$, where $G_1 = L(m,n; k,l+c)$ (resp., $G_1 = L(n,m; l+d,k)$) (see Fig. \emph{\ref{fig:FC13-FC15}(a)--(b)}, and refer to Case III.2.1).
  \item[$\mathrm{(FC14)}$] If $a=2$, $s_y= t_y$, and $c+1\leqslant s_y,t_y\leqslant c+l$, then the length of any path between $s$ and $t$ cannot exceed $\max\{\hat{L}(G_{11}, s, t), \hat{L}(G_{12}, s, t)\}$, where $G_{11} = L(m,n'; k,l')$, $G_{12} = L(s_y,m; k',k)$, $n'=n-s_y+1$, $l'=n'-d$, and $k'=s_y-1$ (see Fig. \emph{\ref{fig:FC13-FC15}(c)--(d)}, and refer to Case III.2.2).
  \item[$\mathrm{(FC15)}$] If $d = 2$, $k > 1$, $a < s_x = t_x < m$, and $n-1\leqslant s_y, t_y\leqslant n$, then the length of any path between $s$ and $t$ cannot exceed $\max\{\hat{L}(G_{11}, s, t), \hat{L}(G_{12}, s, t)\}$, where $G_{11} = C(m',n; k',l; c,d)$, $G_{12} = R(m-m'+1,d)$, $m'=s_x$, and $k'=m'-a$  (see Fig. \emph{\ref{fig:FC13-FC15}(e)--(f)}, and refer to Case III.2.3).
  \item[$\mathrm{(FC16)}$] If $a=2$ and ($\{s, t\} = \{(1, 1), (2, 2)\}$ or $\{(1, 2), (2, 1)\}$) (resp., $d = 1$ and $\{s, t\} = \{(1, n), (2, n-1)\}$ or $\{(1, n - 1), (2, n)\}$), then the length of any path between $s$ and $t$ cannot exceed $\hat{L}(G_1, s, t)$, where $G_1 = L(m,n; k,l+c)$ (resp., $G_1 = L(n,m; l+d,k)$) (see Fig. \emph{\ref{fig:FC16-FC18}(a)--(c)}, and refer to Case III.3.1.1 and Case III.3.2.1).
  \item[$\mathrm{(FC17)}$] If $(C(m,n; k,l; c,d), s, t)$ satisfies condition \emph{(F8)}, then the length of any path between $s$ and $t$ cannot exceed $\hat{L}(G_1, s, t)$, where $G_1 = L(m,n; k,l+c)$ (see Fig. \emph{\ref{fig:FC16-FC18}(d)--(f)}, and refer to Case III.3.1.2).
  \item[$\mathrm{(FC18)}$] If $(C(m,n; k,l; c,d), s, t)$ does not satisfy condition \emph{(F1)}, $k > 1$, and $[(a>2)$ or $(a=2,\{s, t\} \neq \{(1, 1), (2, 2)\}$ and $\{(1, 2), (2, 1)\})]$ (resp., $d = 1$ and $\{s, t\} \neq \{(1, n), (2, n-1)\}$ and $\{(1, n - 1), (2, n)\}$), then the length of any path between $s$ and $t$ cannot exceed $\hat{L}(G_1, s, t)+1$, where $G_1 = L(m,n; k,l+c)$ (resp., $G_1 = L(n,m; l+d,k)$) (see Fig. \emph{\ref{fig:FC16-FC18}(g)--(h)}, and refer to Case III.3.2.2).
\end{description}
\end{lem}
\begin{proof}
For (FC11), consider Fig. \ref{fig:FC11-FC12}(a)--(b). There is only one single path between $s$ and $t$ that has the specified. For (FC12), consider Fig. \ref{fig:FC11-FC12}(c) and \ref{fig:FC11-FC12}(d). Since $c = 1$ and $w$ is a cut vertex, it is clear that the length of any path between $s$ and $t$ cannot exceed $\hat{L}(G_1,s,z)+ \hat{L}(G_2,w,t)$.

For (FC13), consider Fig. \ref{fig:FC13-FC15}(a)--(b). In Fig. \ref{fig:FC13-FC15}(a)--(b), $\{s, t\}$ is a vertex cut and hence removing $s$ and $t$ clearly disconnects $C(m,n; k,l; c,d)$ into two components. Thus, a simple path between $s$ and $t$ can only go through one of these components. Therefore, its length cannot exceed the size of the largest component. Since $c=1$, the larger component will be $G_1 = L(m,n; k,l+c)$. For (FC14) and (F15), consider Fig. \ref{fig:FC13-FC15}(c)--(f). The computations of their upper bounds are the same as (FC13), and (FC9) in Lemma \ref{Lemma:FC9-FC10}.

For (FC16), consider Fig. \ref{fig:FC16-FC18}(a)--(c). A simple check shows that the length of any path between $s$ and $t$ cannot exceed $\hat{L}(G_\textrm{a}, s, p)+ \hat{L}(G_\textrm{b}, q, t)=\hat{L}(L(m,n; k,l+c), s, t)$, where $G_\textrm{a}$, $G_\textrm{b}$, $p$, $q$ are defined in Fig. \ref{fig:FC16-FC18}(a)--(c). For (F17), consider Fig. \ref{fig:FC16-FC18}(d)--(f). In Fig. \ref{fig:FC16-FC18}(d)--(f), let $r = (a,1)$ and $z = (a,2)$, where $r, z$ may be one of $s$ and $t$. Removing $r$ and $z$ clearly disconnects $C(m,n; k,l; c,d)$ into two components and a simple path between $s$ and $t$ can only go through a component that contains $s$, $t$, $r$, and $z$. Since the one disjoint component contains only one vertex, the upper bound of the longest $(s, t)$-path will be $\hat{L}(L(m,n; k,l+c), s, t)$. For (F18), consider Fig. \ref{fig:FC16-FC18}(g)--(h). Since $w$ is a cut vertex, we can easily show that the length of any path between $s$ and $t$ cannot exceed $\hat{L}(G_1, s, t)+1$, where $G_1 = L(m,n; k,l+c)$ or $L(n,m; l+d,k)$.
\end{proof}

Let condition (C1) be defined as follows:

\begin{description}
  \item[$\mathrm{(C1)}$] $(C(m,n;k,l;c,d),s,t)$ does not satisfy any of conditions (F1), (F3), (F7), (F8), and (F9).
\end{description}

It is easy to show that any $(C(m,n;k,l;c,d), s, t)$ must satisfy one of conditions (C1), (FC7), (FC8), (FC9), (FC10), (FC11), (FC12), (FC13), (FC14), (FC15), (FC16), (FC17), and (FC18). If $(C(m,n; k,l; c,d),s , t)$ satisfies (C1), then $\hat{U}(C(m,n; k,l; c,d), s, t) = mn-kl$. Otherwise,  $\hat{U}(C(m,n; k,l; c,d), s, t)$ can be computed using Lemma \ref{Lemma:FC7-FC8}--\ref{Lemma:FC11-FC18}. We summarize them as follows:\\

$\hat{U}(C(m,n; k,l; c,d), s, t)=
  \begin{cases}
   % \hat{L}(G', s, t),                                      &\mathrm{if \ (FC7) \ holds;} \\
    \hat{L}(G_1, s, t),                                     &\mathrm{if \ (FC7),\ (FC13), \ (FC16), \ or \ (FC17) \ holds;} \\
    \hat{L}(G_1, s, z) + \hat{L}(G_2, w, t),                &\mathrm{if \ (FC8) \ or \ (FC12) \ holds;}\\
    \max\{\hat{L}(G_{11}, s, t), \hat{L}(G_{12}, s, t)\},   &\mathrm{if \ (FC9), \ (FC14), \ or \ (FC15) \ holds;} \\
    \hat{L}(G_1, s, t) + k\times c,                         &\mathrm{if \ (FC10) \ holds;} \\
    t_x-s_x+1,                                              &\mathrm{if \ (FC11) \ holds;} \\
    \hat{L}(G_1,s,t)+1,                                     &\mathrm{if \ (FC18) \ holds;} \\
    mn-kl,                                                  &\mathrm{if \ (C1) \ holds.}\\
  \end{cases}$\\

Now, we show how to obtain a longest $(s,t)$-path for $C$-shaped supergrid graphs. Notice that if $(C(m,n; k,l; c,d), s, t)$ satisfies (C1), then, by Theorem \ref{HP-Theorem-Cshaped}, it contains a Hamiltonian $(s, t)$-path.

\begin{lem}\label{Lemma:long-csupergrid}
If $(C(m,n; k,l; c,d), s, t)$ satisfies one of the conditions $\mathrm{(FC7)}$--$\mathrm{(FC18)}$, then $\hat{L}(C(m,n; k,l; c,d), s, t) = \hat{U}(C(m,n; k,l; c,d), s,t)$.
\end{lem}
\begin{proof}
We prove this lemma by constructing a $(s, t)$-path $P$ such that its length equals to $\hat{U}(C(m,n; k,l; c,d), s,t)$. Consider the following cases:

Case 1: Condition (FC7), (FC13), (FC16), or (FC17) holds. Then, by Lemma \ref{Lemma:FC7-FC8} (resp. Lemma \ref{Lemma:FC11-FC18}), $\hat{U}(C(m,n; k,l; c,d),s,t)=\hat{L}(G_1, s, t)$, where $G_1= L(m,n-c; k,l)$ or $G_1 = L(c+1,m; 1,k)$ (resp. $G_1 = L(m,n; k,l+c)$ or $G_1 = L(n, m; l+d,k)$) (see Fig. \ref{fig:FC7-FC8}(a)--(b), Fig. \ref{fig:FC13-FC15}(a)--(b), and Fig. \ref{fig:FC16-FC18}(a)--(f)). Since $G_1$ is a $L$-shaped supergrid graph, by the algorithm of \cite{Keshavarz19a}, we can construct a longest path between $s$ and $t$ in $G_1$.

%Case 2: Condition (FC13), (FC16), or (FC17) holds. Then, by Lemma \ref{Lemma:FC11-FC18}, $\hat{U}(C(m,n; k,l; c,d),s,t)=\hat{L}(G_1, s, t)$, where $G_1 = L(m,n; k,l+c)$ or $G_1 = L(n, m; l+d,k)$ (see Fig. \ref{fig:FC13-FC15}(a)--(b) and Fig. \ref{fig:FC16-FC18}(a)--(f)). Since $G_1$ is a $L$-shaped supergrid graph, by the algorithm of \cite{Keshavarz19a}, we can construct a longest path between $s$ and $t$ in $G_1$.

Case 2: Condition (FC8) or (FC12) holds. By Lemma \ref{Lemma:FC7-FC8} (resp. Lemma \ref{Lemma:FC11-FC18}), $\hat{U}(C(m,n; k,l; c,d), s, t)= \hat{L}(G_1, s, z) + \hat{L}(G_2, w, t)$ (see Fig. \ref{fig:FC7-FC8}(c)--(d) and Fig. \ref{fig:FC11-FC12}(c)--(d)), where $G_1 = L(m,n-c; k,l)$,  $G_2 = R(m, c)$ (resp. $G_1 = L(m,n; k,l+c)$ and $G_2 = R(k, c)$), and $z \thicksim w$. Then, $G_1$ and $G_2$ are $L$-shaped and rectangular supergrid graphs, respectively. First, by the algorithms \cite{Hung17a} and \cite{Keshavarz19a}, we can construct a longest $(s, z)$-path $P_2$ (resp. $P_1$) in $G_2$ (resp. $G_1$) and a longest $(w, t)$-path $P_1$ (resp. $P_2$) in $G_1$ (resp. $G_2$), respectively. Then, $P = P_2\Rightarrow P_1$ (resp. $P_1\Rightarrow P_2$) forms a longest $(s,t)$-path of $C(m,n; k,l; c,d)$. Fig. \ref{fig:FC7-FC8}(c)--(d) and Fig. \ref{fig:FC11-FC12}(c)--(d)) show the constructions of such a longest $(s,t)$-path.

Case 3: Condition (FC9), (FC14), or (FC15) holds. Assume that (FC9) holds. By Lemma \ref{Lemma:FC9-FC10}, $\hat{U}(C(m,n; k,l; c,d), s, t)$ $=$ $\max\{\hat{L}(G_{11}, s, t), \hat{L}(G_{12}, s, t)\}$, where $G_{11} = L(m,n'; k,l')$, $G_{12} = L(s_y,m; k',k)$, $n'=n-s_y+1$, $l'=n'-d$, and $k'=s_y-1$ (see Fig. \ref{fig:FC9-FC10}(a)--(b)). Since $G_{11}$ and $G_{12}$ are $L$-shaped supergrid graphs, by the algorithm of \cite{Keshavarz19a} we can construct a longest path between $s$ and $t$ in $G_{11}$ and $G_{12}$. Fig. \ref{fig:FC9-FC10}(a)--(b) depict such a construction. For conditions (FC14) and (FC15), consider Fig. \ref{fig:FC13-FC15}(c)--(f). Then, $G_{12}$ may be a rectangle. By the algorithm of \cite{Hung17a} we can construct a longest path between $s$ and $t$ in $G_{12}$ if it is a rectangle. In addition, $G_{11}$ is a $C$-shaped supergrid graph in (FC15) (see Fig. \ref{fig:FC13-FC15}(e). Then, $(G_{11}, s, t)$ satisfies condition (FC18). And its longest $(s, t)$-path can be computed by the algorithm in \cite{Keshavarz19a}. Its construction is shown in Case 6 and Fig. \ref{fig:FC13-FC15}(e) shows such a construction of a longest $(s, t)$-path.

Case 4: Condition (FC10) holds. By Lemma \ref{Lemma:FC9-FC10}, $\hat{U}(C(m,n; k,l; c,d), s, t) = \hat{L}(G_1, s, t)+|G_2|=\hat{L}(G_1, s, t)+ k\times c$, where $G_1 = L(m,n; k,l+c)$ and $G_2 = R(k, c)$. Consider Fig. \ref{fig:FC9-FC10}(c)--(d). Then, $G_1$ is a $L$-shaped supergrid graph and $G_2$ is a rectangle. By the algorithm of \cite{Keshavarz19a}, we can construct a longest path $P_1$ between $s$ and $t$ in $G_1$ that contains edge $e_1$ locating to face $G_2$. By Lemma \ref{HC-rectangular_supergrid_graphs}, $G_2$ contains canonical Hamiltonian cycle $HC_2$ such that its one flat face is placed to face $G_1$. Thus, by Statement (2) of Proposition \ref{Pro_Obs}, $P_1$ and $HC_2$ can be combined into a longest $(s, t)$-path of $C(m,n; k,l; c,d)$.

Case 5: Condition (FC11) holds. By Lemma \ref{Lemma:FC11-FC18}, $\hat{U}(C(m,n; k,l; c,d), s, t) = t_x-s_x+1$. Obviously, the lemma holds for the single possible path between $s$ and $t$ (see Fig. \ref{fig:FC11-FC12}(a)--(b)).

Case 6: Condition (FC18) holds. By Lemma \ref{Lemma:FC11-FC18}, $\hat{U}(C(m,n; k,l; c,d), s, t) = \hat{L}(G_1, s, t)+1$, where $G_1 = L(m,n; k,l+c)$. We make a vertical and horizontal separations on $C(m,n; k,l; c,d)$ to obtain two disjoint supergrid subgraphs $R_2=L(n_2, m_2; 1,1)$ and $R_1 = R(m, n-n_2)$, where $n_2=c+l$ and $m_2=a+1$ (see Fig. \ref{fig:FC18}(a)). Note that $(C(m,n; k,l; c,d), s, t)$ does not satisfy condition (F1) in this case. Depending on the positions of $s$ and $t$, there are the following three subcases:

\begin{figure}[h]
\centering
\includegraphics[scale=1]{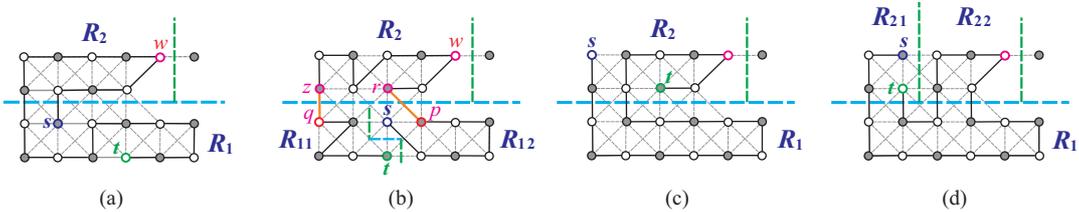}
\caption[]{A longest $(s,t)$-path in $C(m,n; k,l; c,d)$, where (FC18) holds.}
\label{fig:FC18}
\end{figure}

\hspace{0.5cm}Case 6.1: $s,t\in R_1$. A longest $(s, t)$-path of $(C(m, n; k, l; c,d)$ can be constructed by similar arguments in proving Case 1 of Lemma \ref{HP-Cshaped3} (see Fig. \ref{fig:FC18}(a)--(b)). Then, we can construct a Hamiltonian $(s, t)$-path of $R_1\cup R_2$. Fig. \ref{fig:FC18}(a) and Fig. \ref{fig:FC18}(b) depict such constructions. The size of constructed Hamiltonian $(s, t)$-path equals to $\hat{L}(L(m,n; k,l+c))+1=|V(R_1\cup R_2)|$, and hence it is the longest $(s, t)$-path of $(C(m, n; k, l; c,d)$.

\hspace{0.5cm}Case 6.2: $s,t\in R_2$. A longest $(s, t)$-path of $(C(m,n; k,l; c,d)$ can be constructed by similar arguments in proving Case 2 of Lemma \ref{HP-Cshaped3} (see Fig. \ref{fig:FC18}(c)--(d)). Depending on whether $\{s, t\}$ is a vertex cut of $R_2$, we consider Fig. \ref{fig:FC18}(c) and Fig. \ref{fig:FC18}(d). Then, we can construct a Hamiltonian $(s, t)$-path of $R_1\cup R_2$, as shown in Fig. \ref{fig:FC18}(a) and Fig. \ref{fig:FC18}(b). The size of constructed Hamiltonian $(s, t)$-path equals to $\hat{L}(L(m,n; k,l+c))+1=|V(R_1\cup R_2)|$, and hence it is the longest $(s, t)$-path of $(C(m, n; k, l; c,d)$.

\hspace{0.5cm}Case 6.3: $(s\in R_1$ and $t\in R_2)$ or $(s\in R_2$ and $t\in R_1)$. Without loos of generality, assume that $s\in R_1$ and $t\in R_2$. A longest $(s, t)$-path of $(C(m,n; k,l; c,d)$ can be constructed by similar arguments in proving Case 2.2 of Lemma \ref{HP-Cshaped2} and Case 3 of Lemma \ref{HP-Cshaped3}.
\end{proof}

It follows from Theorem \ref{HP-Theorem-Cshaped} and Lemmas \ref{Lemma:FC7-FC8}--\ref{Lemma:long-csupergrid} that the following theorem concludes the result.

\begin{thm}\label{LongPathC}
Given a $C$-shaped supergrid $C(m,n; k,l; c,d)$ and two distinct vertices $s$ and $t$ in $C(m,n; k,l; c,d)$, a longest $(s, t)$-path can be found in $O(mn)$-linear time.
\end{thm}

The linear-time algorithm is formally presented as Algorithm \ref{TheHamiltonianPathAlgm}.

\begin{algorithm}[tb]
  \SetCommentSty{small}
  \LinesNumbered
  \SetNlSty{textmd}{}{.}

    \KwIn{A $C$-shaped supergrid graph $C(m,n;k,l;c,d)$ with $mn\geqslant 2$, and two distinct vertices $s$ and $t$ in $C(m,n;k,l;c,d)$.}
    \KwOut{The longest $(s, t)$-path.}

\textbf{if} $a(=m-k)=1$ \textbf{then} \textbf{output} $HP(C(m, n;k,l;c,d), s, t)$ constructed from Lemma \ref{HP-Cshaped1}; // $(C(m, n;k,l;c,d), s, t)$ does not satisfy the forbidden conditions (F1), (F3), (F7), (F8), and (F9);\\
\textbf{if} $a(=m-k)>1$ and $c=d=1$ \textbf{then} \textbf{output} $HP(C(m, n;k,l;c,d), s, t)$ constructed from Lemma \ref{HP-Cshaped2}; // $(C(m, n;k,l;c,d), s, t)$ does not satisfy the forbidden conditions (F1), (F3), (F7), and (F8);\\
\textbf{if} $a(=m-k)>1$ and $c>1$ or $d>1$ \textbf{then} \textbf{output} $HP(C(m, n;k,l;c,d), s, t)$ constructed from Lemma \ref{HP-Cshaped3}; // $(C(m, n;k,l;c,d), s, t)$ does not satisfy the forbidden conditions (F1), (F3), (F7), and (F8);\\
\textbf{if} $(C(m, n;k,l;c,d), s, t)$ satisfies one of the forbidden conditions (F1), (F3), (F7), (F8), and (F9),  \textbf{then} \textbf{output} the longest $(s, t)$-path based on Lemma \ref{Lemma:long-csupergrid};\\
\caption{The longest $(s, t)$-path algorithm}
\label{TheHamiltonianPathAlgm}
\end{algorithm}

%====================================================================
\section{Concluding remarks}\label{Sec_Conclusion}
%====================================================================
Based on the Hamiltonicity and Hamiltonian connectivity of rectangular supergrid graphs, we first discover some Hamiltonian connected properties of rectangular supergrid graphs. Using the Hamiltonicity and Hamiltonian connectivity of rectangular and $L$-shaped supergrid graphs, we then prove $C$-shaped supergrid graphs to be Hamiltonian connected except few conditions. On the other hand, the Hamiltonian cycle problem on solid grid graphs was known to be polynomial solvable. However, it remains open for solid supergrid graphs in which there exists no hole. We leave it to interesting readers.

\section*{Acknowledgments}
%The authors would like to thank anonymous referees for many useful comments and suggestions which have improved the presentation and correctness of this paper.
This work is partly supported by the Ministry of Science and Technology, Taiwan under grant no. MOST 105-2221-E-324-010-MY3.

\end{document}